\documentclass[11pt]{article}

\usepackage{amsmath}
\usepackage{amssymb}
\usepackage{amsthm}
\usepackage{latexsym}

\usepackage{float}
\usepackage{array}
\usepackage{xcolor}
\usepackage{enumerate} 
\usepackage[shortlabels]{enumitem}
\usepackage{multirow}
\usepackage[hyphens]{url}
\usepackage{hyperref}
\usepackage{etoolbox}
\usepackage{tcolorbox}
\usepackage{MnSymbol}
\usepackage{ifthen}
\usepackage{mathrsfs}
\usepackage{graphicx,color}

\newtheoremstyle{mystyle}%
{3pt}
{3pt}
{\color{blue}}
{}
{\bfseries\color{blue}}
{.}
{.5em}
{}

\theoremstyle{plain}
\newtheorem{theorem}{Theorem}[section]

\newtheorem{prop}[theorem]{Proposition}
\newtheorem{proposition}[theorem]{Proposition}

\newtheorem{lemma}[theorem]{Lemma}
\newtheorem{claim}[theorem]{Claim}

\theoremstyle{definition}

\theoremstyle{remark}
\newtheorem{remark}[theorem]{Remark}

\theoremstyle{mystyle}

\mathchardef\mhyphen="2D
\newcommand{\cl}[1]{\mathscr{#1}}

\newcommand{\eat}[1]{}

\newcommand{\mso}{\ensuremath{\mathrm{MSO}}}
\newcommand{\fo}{\ensuremath{\mathrm{FO}}}

\newcommand{\hucw}{\ensuremath{\mathsf{HUCW}}}
\renewcommand{\top}[1]{\ensuremath{\mathsf{top}(#1)}}
\newcommand{\topsucc}[1]{\ensuremath{\mathsf{topsucc}(#1)}}
\newcommand{\bottom}[1]{\ensuremath{\mathsf{bottom}(#1)}}
\newcommand{\penult}[1]{\ensuremath{\mathsf{penult}(#1)}}
\newcommand{\prepenult}[1]{\ensuremath{\mathsf{prepenult}(#1)}}
\newcommand{\rlboundary}[1]{\ensuremath{\mathsf{H}\text{-}\mathsf{boundary}(#1)}}
\newcommand{\one}[1]{\ensuremath{\mathsf{one}(#1)}}

\newcommand{\odd}[1]{\ensuremath{\mathsf{odd}(#1)}}
\newcommand{\first}[1]{\ensuremath{\mathsf{first}(#1)}}
\newcommand{\last}[1]{\ensuremath{\mathsf{last}(#1)}}
\newcommand{\samecol}[1]{\ensuremath{\mathsf{samecolumn}(#1)}}
\newcommand{\domain}[1]{\ensuremath{\mathsf{domain}(#1)}}
\newcommand{\adjcol}[1]{\ensuremath{\mathsf{adjcolumn}(#1)}}

\newcommand{\linord}[1]{\ensuremath{\mathsf{linord}(#1)}}
\newcommand{\hedge}[1]{\ensuremath{\mathsf{H}\text{-}\textsf{edge}(#1)}}
\newcommand{\vedge}[1]{\ensuremath{\mathsf{V}\text{-}\textsf{edge}(#1)}}
\newcommand{\prepenultedge}[1]{\ensuremath{\mathsf{prepenultedge}(#1)}}
\renewcommand{\path}[1]{\ensuremath{\mathsf{path}(#1)}}

\newcommand{\col}[1]{
  \ifthenelse{\equal{#1}{-1}}{\ensuremath{\mathrm{C}}}{\ensuremath{\mathrm{C}_{#1}}}}
\newcommand{\row}[1]{
  \ifthenelse{\equal{#1}{-1}}{\ensuremath{\mathrm{R}}}{\ensuremath{\mathrm{R}_{#1}}}}
\newcommand{\mycolor}[2]{\ensuremath{\mathsf{Colour}_{#1}(#2)}}
\newcommand{\rhscol}[2]{\ensuremath{\mathsf{rhscolumn}[#1](#2)}}
\newcommand{\lessthan}[1]{\ensuremath{\mathsf{lessthan}(#1)}}
\newcommand{\powergraphs}{\ensuremath{\mathsf{Power}\text{-}\mathsf{graphs}}}

\newcommand{\ctwomso}{\ensuremath{\mathrm{C}_2\mso}}

\makeatletter
\newsavebox\myboxA
\newsavebox\myboxB
\newlength\mylenA

\newcommand*\xoverline[2][0.75]{%
    \sbox{\myboxA}{$\m@th#2$}%
    \setbox\myboxB\null
    \ht\myboxB=\ht\myboxA%
    \dp\myboxB=\dp\myboxA%
    \wd\myboxB=#1\wd\myboxA
    \sbox\myboxB{$\m@th\overline{\copy\myboxB}$}
    \setlength\mylenA{\the\wd\myboxA}
    \addtolength\mylenA{-\the\wd\myboxB}%
    \ifdim\wd\myboxB<\wd\myboxA%
       \rlap{\hskip 0.5\mylenA\usebox\myboxB}{\usebox\myboxA}%
    \else
        \hskip -0.5\mylenA\rlap{\usebox\myboxA}{\hskip 0.5\mylenA\usebox\myboxB}%
    \fi}
    
  \newcommand*\xunderline[2][0.75]{%
    \sbox{\myboxA}{$\m@th#2$}%
    \setbox\myboxB\null
    \ht\myboxB=\ht\myboxA%
    \dp\myboxB=\dp\myboxA%
    \wd\myboxB=#1\wd\myboxA
    \sbox\myboxB{$\m@th\underline{\copy\myboxB}$}
    \setlength\mylenA{\the\wd\myboxA}
    \addtolength\mylenA{-\the\wd\myboxB}%
    \ifdim\wd\myboxB<\wd\myboxA%
       \rlap{\hskip 0.5\mylenA\usebox\myboxB}{\usebox\myboxA}%
    \else
        \hskip -0.5\mylenA\rlap{\usebox\myboxA}{\hskip 0.5\mylenA\usebox\myboxB}%
    \fi}
\makeatother

\newcommand{\cw}{\text{cwd}}
\newcommand{\tw}{\text{twd}}
\renewcommand{\bigcupdot}{\biguplus}

\newcommand{\tbf}[1]{\textbf{#1}}
\newcommand{\mc}[1]{\mathcal{#1}}

\newcommand{\true}{\ensuremath{\mathsf{True}}}

\usepackage{lineno,hyperref,authblk}    
\begin{document}

\title{\tbf{MSO Undecidability for Hereditary Classes of Unbounded Clique-Width}}
\author[1]{Anuj Dawar} 
\author[2]{Abhisekh Sankaran}
\affil[1]{Department of Computer Science and Technology, \authorcr University of Cambridge, U.K. \authorcr \texttt{anuj.dawar@cl.cam.ac.uk}}
\affil[ ]{\phantom{text}}
\affil[2]{Tata Consultancy Services Ltd., Mumbai, India \authorcr \texttt{abhisekh.sankaran@tcs.com}}
\date{}                     
\renewcommand\Affilfont{\small}
\maketitle

\begin{abstract}
Seese's conjecture for finite graphs states that monadic second-order logic (MSO) is undecidable on all graph classes of unbounded clique-width.  We show that to establish this it would suffice to show that grids of unbounded size can be interpreted in two families of graph classes: minimal hereditary classes of unbounded clique-width; and antichains of unbounded clique-width under the induced subgraph relation.  We explore all the currently known classes of the former category and establish that grids of unbounded size can indeed be interpreted in them.
\end{abstract}


\section{Introduction}\label{section:intro}

The monadic second-order logic ($\mso$) of graphs has been an object of intensive research for many years now.  It is a logic that is highly expressive and yet very well behaved on many interesting classes of graphs.  It has enabled the extension of many automata-theoretic and algebraic techniques to the construction of algorithms on graphs (see the comprehensive treatment in~\cite{courcelle-engelfriet}).  It has become a reference logic against which many others are compared.  A key area of investigation is determining on which classes of graphs  $\mso$ is algorithmically well-behaved.

The good algorithmic behaviour of $\mso$ on a class $\cl{C}$ of graphs is usually taken to mean one of two things: the evaluation (or model-checking) problem for $\mso$ sentences on $\cl{C}$ is tractable; or the satisfiability problem of $\mso$ sentences on $\cl{C}$ is decidable.  Usually, these two are linked.  Broadly speaking, the only way we know to show that the $\mso$ theory of a class $\cl{C}$ is decidable is to show that $\cl{C}$ can be obtained by means of an $\mso$ interpretation from a class of trees, which itself has a decidable theory and this also yields efficient evaluation algorithms for $\mso$ sentences on $\cl{C}$.  And the only way we know to show that the $\mso$ theory of $\cl{C}$ is undecidable is to show that there is an $\mso$ interpretation that yields arbitrarily large grids on $\cl{C}$ and, often this also yields an obstacle to the tractability of $\mso$ evaluation on $\cl{C}$.  There are exceptions to the latter in pathological cases (for instance, if the interpreted grid is much smaller than the structure in which it is interpreted) but~\cite{Kreutzer09} provides a fairly general instance of the rule.

Seese~\cite{Seese} formalizes the first of these observations into a conjecture: if the $\mso$ theory of a class $\cl{C}$ is decidable, there is an $\mso$ interpretation $\Psi$ and a class $\cl{T}$ of trees such that $\Psi$ maps $\cl{T}$ to $\cl{C}$.  This  remains an open question nearly three decades after it was first posed despite considerable research effort around it.  By a theorem of Courcelle and Engelfriet~\cite{courcelle-engelfriet}, it is known that the classes of graphs obtained by $\mso$ interpretations from trees are exactly those of bounded clique-width.  Thus, Seese's conjecture can be understood as saying that any class of graphs of unbounded clique-width has an undecidable $\mso$ theory.  If we similarly formalize the second observation above about grids and combine it with this, we can formulate the following stronger conjecture, which we refer to below as the \emph{strong Seese conjecture}: every class $\cl{C}$ of graphs of unbounded clique-width admits an $\mso$ interpretation that defines arbitrarily large grids.  Seese's conjecture is often formulated in this stronger form as it seems the only reasonable route to proving it. 
It can be seen as an interesting analogue of the Robertson-Seymour grid minor theorem to the effect that any class of graphs of unbounded treewidth admits arbitrarily large grids as minors.

In recent years there has been growing interest in clique-width as a measure of the complexity of graphs from a structural and algorithmic point of view, quite separate from questions of logic~\cite{courcelle-olariu, corneil-rotics, oum-seymour, GH011}.  In particular, it provides a route for extending algorithmic methods that have had great success on sparse graph classes~\cite{Sparsity} to more general classes of graphs.  A class of graphs may be of bounded clique-width while containing dense graphs---the classic example being the class of cliques.

In the context of the structural study of classes of bounded clique-width, there is particular interest in \emph{hereditary classes}, that is, classes of graphs closed under the operation of taking induced subgraphs.  This is because the induced subgraph relation behaves well with respect to clique-width.  If a graph $H$ is a subgraph or a minor of a graph $G$, the clique-width of $H$ can be greater than that of $G$ but if $H$ is an \emph{induced subgraph} of $G$, then the clique-width of $H$ is no more than that of $G$.  Hence, the hereditary closure of a class $\cl{C}$ of bounded clique-width still has bounded clique-width.

The induced subgraph relation is not as well-behaved as the graph minor relation.  By the Robertson-Seymour graph minor theorem~\cite{RobertsonS-GM20}, the graph minor relation is a well-quasi-order.  This is not true for the induced subgraph relation.
By the same token, the classes of graphs of unbounded treewidth are well understood  in that they are precisely the classes which have grid minors of unbounded size.  The picture for classes of graphs of unbounded clique-width is somewhat less clear.  In particular, the relationship between a class having unbounded clique-width and admitting a well-quasi-order of the induced subgraph relation has been the subject of much investigation.   It is possible to construct as we see below, infinite descending chains, under inclusion, of hereditary classes of graphs, each of unbounded clique-width.

Lozin~\cite{Lozin11} identified the first example of a  hereditary class $\cl{C}$  of graphs of unbounded clique-width that are \emph{minimal} with this property---that is, no hereditary proper subclass of $\cl{C}$ has unbounded clique-width.  Since then, many other such classes have been constructed.  Collins et al.~\cite{CollinsFKLZ18} show how to obtain an infinite family of such classes.  Their construction has been recently extended by Brignall and Cocks~\cite{BC22} to obtain uncountably many examples.  Atminas et al.~\cite{ABLS15} construct instances of such classes which are characterized by a finite collection of forbidden induced subgraphs.  Lozin et al.~\cite{LRZ15} construct a minimal hereditary class of unbounded clique-width that is well-quasi-ordered under the induced subgraph relation.

This exploration of novel classes of unbounded clique-width also suggests an approach to establishing Seese's conjecture for 
finite graphs.  We establish in Section~\ref{sec:wqo-minimal} that a proof of Seese's conjecture would follow from the conjunction of the following two 
statements: (1) every collection of graphs of unbounded clique-width that forms an infinite anti-chain under the induced subgraph 
relation interprets arbitrarily large grids; and (2) every minimal hereditary class of unbounded clique-width interprets arbitrarily 
large grids.  This suggests a programme to establish Seese's conjecture by systematically studying antichains and minimal hereditary 
classes of unbounded clique-width.  We do not yet know of a complete classification of minimal hereditary classes of unbounded 
clique-width, which makes a systematic approach to this programme challenging.  Nevertheless, we examine in Sections~\ref{sec:grid-classes}--\ref{section:power-graphs} all known classes satisfying these conditions and show that in all cases we can indeed interpret grids of unbounded size.  Thus none of these provides a counterexample to Seese's conjecture.
Our construction shows a uniform method of proving that these classes have unbounded clique-width.  The proof is often simpler than the \emph{ad hoc} methods by which this was proved for each class in the literature.

It is worth mentioning some significant lines of investigation related to Seese's conjecture.  Courcelle~\cite{courcelle} shows that 
proving Seese's conjecture for finite graphs is equivalent to proving the relativized version of the conjecture for particular classes 
of graphs, two examples being bipartite graphs and split graphs.  He further shows the conjecture to be true when relativized to 
uniformly $k$-sparse graphs and interval graphs.  Another line of work addresses variants of Seese's conjecture obtained by considering 
logics other than $\mso$.  One such result by Seese~\cite{Seese} shows that guarded second-order logic (GSO) is undecidable on any 
class of unbounded clique-width.  Similarly, Courcelle and Oum~\cite{courcelle-oum} show that the extension $\ctwomso$ of $\mso$ 
obtained by considering modulo 2 counting quantifiers is also undecidable on classes of unbounded clique-width. In all of these cases, 
the proof goes via interpreting grids in unbounded clique-width classes.  There has also been interesting progress looking at Seese's 
conjecture for structures other than graphs.  A significant paper here is by Hlin{\v{e}}n{\'y} and Seese~\cite{HS06} who show the 
conjecture to be true for matroids representable over any finite field.

\section{Preliminaries}\label{section:prelims}
The graphs we consider in this paper are simple, undirected and loop-free. For a graph $G$, we write $V(G)$ for the vertices of $G$ and $E(G)$ for the edges. A graph $H$ is an \emph{induced subgraph} of $G$ if $V(H) \subseteq V(G)$ and for any $x,y \in V(H)$, $\{x,y\} \in E(H)$ if, and only if, $\{x,y\} \in E(G)$.  We write $H \subseteq G$ to denote that~$H$ is an induced subgraph of $G$. A class of graphs is said to be \emph{hereditary} if it is closed under induced subgraphs.  For any graph class $\cl{C}$, we write $\cl{C}\!\downarrow$ to denote the hereditary closure of $\cl{C}$ -- that is, the class of graphs that are induced subgraphs of the graphs in $\cl{C}$.  We consider monadic second-order logic ($\mso$) over vocabularies $\tau$ containing the binary relation $E$ and finitely many unary relation symbols.  A $\tau$-structure can be thought of as an \emph{expansion} of a graph $G = (V,E)$ with unary relations that interpret the unary symbols in~$\tau$.  Such a structure can be thought of as a vertex-coloured graph.
An $\mso$ formula over the vocabulary $\tau$ is an expression that is inductively constructed from atomic $\mso$ formulae using the Boolean connectives $\wedge, \vee,$ and $\neg$, and existential quantification over vertex variables and set variables.  Here an atomic $\mso$ formula is an expression of the form $E(x, y)$ or $Q(x)$ or $X(y)$ or $x = y$ where~$x,~y$ are vertex variables, the predicates $E, Q$ belong to $\tau$ and $X$ is a set variable. A \emph{first order}, or $\fo$, formula is an $\mso$ formula that does not contain any set variable.  We often write $\varphi(\bar{x}, \bar{X})$ to denote a formula whose free variables are among $\bar{x}$ and $\bar{X}$, the former being a tuple of vertex variables and the latter a tuple of set variables.  Given such a formula, and a graph $G$ along with a tuple $\bar{a}$ of vertices that interprets~$\bar{x}$ and a tuple $\bar{A}$ of unary relations that interprets~$\bar{X}$, we write $(G,\bar{A}) \models \phi[\bar{a}]$ to denote that the formula $\phi$ is satisfied in~$G$ in this interpretation.

Given a graph $G$ and an $\mso$ formula $\varphi(\bar{x}, \bar{X})$ where the length of $\bar{x}$ is~$k$, we can think of $\varphi$ as defining a $k$-ary relation on an expansion of $G$ with an interpretation $\bar{A}$ of $\bar{X}$. Specifically this relation, denoted $\varphi^{(G, \bar{A})}$, is given by $\varphi^{(G, \bar{A})} = \{\bar{a} \mid (G,\bar{A}) \models \varphi[\bar{a}]\}$. Given a sequence $\bar{Z}$ of set variables, an \emph{$\mso$ graph interpretation with parameters $\bar{Z}$} is a pair $\Psi(\bar{Z}) = (\psi_V(x, \bar{Z}), \psi_E(x, y, \bar{Z}))$ of $\mso$ formulas over the vocabulary $\{E\} \cup \{Z_i \mid Z_i~\text{is an element of}~\bar{Z}\}$.  Given a graph $G$ together with unary relations $\bar{A}$ interpreting the set variables $\bar{Z}$ in $G$, the interpretation $\Psi(\bar{Z})$ defines a possibly directed graph $H = \Psi(G, \bar{A})$.  This graph has (i) vertex set $\psi_V^{(G, \bar{A})}$, and (ii) edge set  $\psi_E^{(G, \bar{A})}$. In this paper, we are only interested in the case where $\Psi(\bar{Z})$ defines an undirected graph (that is,~$\psi_E(x, y, \bar{Z})$ defines an irreflexive and symmetric binary relation). Thus $\Psi(\bar{Z})$ defines a function from the expansions of graphs with $|\bar{Z}|$ unary predicates, to graphs, and therefore in general defines a relation on graphs.  Where it causes no confusion, we also refer to the relation defined by an interpretation as an $\mso$ interpretation.
If $\bar{Z}$ is empty, we call the interpretation $\Psi$ \emph{parameterless}, and such a $\Psi$ defines a function from graphs to graphs.
An example of a parameterless interpretation is $\Theta = (\theta_V(x), \theta_E(x, y))$ where $\theta_V(x) := (x = x)$ and $\theta_E(x, y) :=  \neg E(x, y)$; the function it defines maps a graph to its complement.  An example of an interpretation with parameters is $\Gamma(Z) = (\gamma_V(x, Z), \gamma_E(x, y, Z))$ where $\gamma_V(x, Z) := Z(x)$ and $\gamma_E(x, y, Z) := E(x, y)$. The function that it defines maps an expansion $(G, A)$ of a graph $G$ to the subgraph of $G$ induced by $A$; thus the relation on graphs that $\Gamma(Z)$ defines maps a graph to all its induced subgraphs.
Given a class $\cl{C}$ of graphs and an interpretation $\Psi$ with parameters~$\bar{Z}$, we denote by $\Psi(\cl{C})$ the class of graphs given by $\Psi(\cl{C}) = \{\Psi(G, \bar{A}) \mid G \in~\cl{C}~\mbox{and}~\bar{A}~\mbox{is an interpretation of}~\bar{Z}~\mbox{in}~G\}$.  For example, for the interpretation $\Gamma$ above and a class $\cl{C}$ of graphs, the class $\Gamma(\cl{C})$ is exactly the hereditary closure of $\cl{C}$. Since they are relations, one can compose interpretations and it is known that the class of $\mso$ interpretations is closed under composition~\cite{hodges}. We call  $\mso$ interpretations with parameters simply $\mso$ interpretations for ease of readability, and denote them with the uppercase Greek letters $\Phi, \Gamma, \Psi, \Theta$, etc.

The notion of clique-width is a structural parameter of graphs that was introduced by Courcelle, Engelfriet and Rozenberg in~\cite{cer93} as a generalization of the well-known notion of treewidth.  We do not give the definitions of clique-width and treewidth here as we need only specific properties of these for our results that we state below; we point the reader to~\cite{courcelle-engelfriet, Sparsity} for more about the notions and results concerning them. We write $\cw(G)$ and $\tw(G)$ for the clique-width and tree-width of a graph $G$,  respectively.  As examples, a clique has clique-width 1, and a cograph has clique-width 2.  It is known for any graph $G$ that $\cw(G) \leq 4 \cdot 2^{\tw(G)-1} + 1$~\cite{courcelle-olariu} and for planar $G$ we even have $\cw(G) \leq 6\tw(G) -2$~\cite{courcelle2018}. A class of graphs is said to  have \emph{bounded} clique-width if for some number $k \ge 1$, every graph in the class has clique-width at most~$k$.  Thus, the class of cliques, the class of  cographs and all classes of bounded treewidth have bounded clique-width. A graph class has \emph{unbounded} clique-width if it does not have bounded clique-width. Examples of graph classes of unbounded clique-width include grids, interval graphs, and line graphs~\cite{courcelle}.

The class of all graphs of clique-width at most $k$ is hereditary since the clique-width of an induced subgraph of $G$ is never more than the clique-width of $G$.  An \emph{antichain} under the induced subgraph relation is a set $\mc{A}$ of graphs such that if $G$ and $H$ are distinct graphs in $\mc{A}$, then neither of $G \subseteq H$ or $H \subseteq G$ holds.  Usually when we say ``antichain'' without further qualification, we mean an antichain under the induced subgraph relation.
A graph class $\cl{C}$ is said to be \emph{well-quasi-ordered} (WQO) under induced subgraphs if it does not contain any infinite antichains. For example, the class of all cliques is WQO under induced subgraphs.

The $\mso$ theory of a graph class $\cl{C}$ is the class of all $\mso$ sentences that are true in all graphs of $\cl{C}$.  This theory is decidable if, and only if, the following problem is decidable: given an $\mso$ sentence $\phi$ decide if $\phi$ is true in some graph in $\cl{C}$.
Seese's conjecture states any class whose $\mso$ theory is decidable has bounded clique-width.  An $m \times n$ grid is a graph $G = (V, E)$ on $m \cdot n$ vertices with $V = \{(i, j) \mid 1 \leq i \leq m, 1 \leq j \leq n\}$ and $E = \{ \{(i, j), (i, j+1)\} \mid 1 \leq i \leq m, 1 \leq j < n\}  \cup \{ \{(i, j), (i+1, j)\} \mid 1 \leq i < m, 1 \leq j \leq n\}$. The grid is \emph{square} if $m = n$. We say a class $\cl{C}$ of graphs \emph{interprets grids} via an $\mso$ interpretation $\Phi$, if~$\Phi(\cl{C})$ contains graphs isomorphic to arbitrarily large square grids.  Any class of graphs that contains arbitrarily large grids has undecidable $\mso$ theory~\cite[Thm.~5.6]{courcelle-engelfriet}.  Morover, since $\mso$ decidability is preserved by interpretations~\cite[Thm.~7.54]{courcelle-engelfriet}, any class of graphs that interprets grids via an $\mso$ interpretation has an undecidable $\mso$ theory.  The \emph{strong Seese conjecture} is that any class of unbounded clique-width interprets grids via an $\mso$ interpretation.
It is known that if the clique-width of a class $\cl{C}$ is bounded and $\Phi$ is an $\mso$ interpretation, then the clique-width of $\Phi(\cl{C})$ is also bounded~\cite[Cor.~7.38]{courcelle-engelfriet}. A simple observation about classes interpreting grids is the following.
\begin{proposition}\label{prop:reduction-helper}
 Suppose $\cl{C}$ is a graph class that interprets grids, and $\cl{D}$ is a graph class for which there exists an MSO interpretation $\Xi$  such that the hereditary closure of $\Xi(\cl{D})$ contains $\cl{C}$. Then $\cl{D}$ interprets grids as well. 
\end{proposition}

Specifically, if $\Theta$ is the interpretation mapping $\cl{C}$ to a class containing arbitrarily large grids, and $\Gamma$ is the interpretation defined above taking any class to its hereditary closure, then 
an interpretation $\Omega$ such that $\Omega(\cl{D})$ contains arbitrarily large square grids, is given by $\Omega = \Theta \circ \Gamma \circ \Xi$ (viewing $\Theta, \Gamma$ and $\Xi$ as functions) where $\circ$ denotes composition.

We say that a class of graphs $\cl{C}$ is $\hucw$ if it is hereditary and has unbounded clique-width.  An $\hucw$ graph class is said to be  \emph{minimal} if it does not contain a proper subclass that is $\hucw$.  For example, bipartite permutation graphs and unit interval graphs are two minimal $\hucw$ graph classes~\cite{Lozin11}.  The existence of countably many minimal $\hucw$ classes is established in~\cite{CollinsFKLZ18}, and this has been recently extended to uncountably many minimal $\hucw$ classes in~\cite{BC22}.

\section{Minimal Classes and Well-Quasi-Ordering}\label{sec:wqo-minimal}

In this section we lay out an approach to studying Seese's conjecture that motivates our study of $\mso$ decidability for minimal $\hucw$ classes. The first observation is that, if $\cl{C}$ is a counter-example to Seese's conjecture, then so is~${\cl{C}\!\downarrow}$.    Recall that a counter-example to Seese's conjecture would be a class~$\cl{C}$ that has unbounded clique-width and a decidable $\mso$ theory.  Clearly if~$\cl{C}$ has unbounded clique-width, then so does  ${\cl{C}\!\downarrow}$.  The following proposition is folklore.  It follows immediately from the fact that $\mso$ decidability is preserved by interpretations and the existence of the interpretation $\Gamma$ defined above which takes a class to its hereditary closure.

\begin{prop}\label{prop:hereditary}
 If the $\mso$ theory of  $\cl{C}$ is decidable, then so is the $\mso$ theory of ${\cl{C}\!\downarrow}$.
\end{prop}

Hence, if there is a counter-example to Seese's conjecture, we have one that is a hereditary class of unbounded clique-width, i.e.\ an $\hucw$ class.  In the present section, we establish some basic facts about the $\hucw$ classes that allow us to structure the search for such a counter-example, or indeed the attempt to show that there is none.  

The relation of being an induced subgraph is not a well-quasi-order as it admits infinite anti-chains.  As an example, let $I_n$ be the graph on $n+4$ vertices $e_0,e_1,e_2,e_3,c_1,\ldots,c_n$ where for each $i < n$ there is an edge between $c_i$ and $c_{i+1}$, and in addition we have edges $e_0-c_1$, $e_1-c_1$, $e_2-c_n$ and $e_3-c_n$.  In short, there is a path of length $n$ with two additional vertices at each end to mark the ends.  Then, it is clear the collection $(I_n)_{n\in \mathbb{N}}$ is an antichain in the induced subgraph order.  This particular antichain has bounded clique-width.   It is also possible to construct antichains of unbounded clique-width (which therefore must be infinite).  An example is obtained by taking the collection of $n \times n$ grids and adding two extra vertices at each corner to form a triangle.  In what follows, whenever we refer to an \emph{antichain} we mean one under the induced subgraph relation.

From an  antichain of unbounded clique-width, it is possible to construct (as we show below) an infinite descending chain of classes of graphs (under the inclusion relation) all of which are $\hucw$.  Thus, it was a significant discovery to find that there are actually $\hucw$ classes $\cl{C}$ that are \emph{minimal}: no proper hereditary subclass of $\cl{C}$ has unbounded clique-width.  The first such example is due to Lozin~\cite{Lozin11}.  Collins et al.~\cite{CollinsFKLZ18} constructed an infinite family of such classes and Lozin et al.~\cite{LRZ15}  give an example that is itself well-quasi-ordered under the induced substructure relation.  We examine these in some detail in subsequent sections.

If it were the case that every class that is $\hucw$ contains as a subclass a minimal $\hucw$ class, then showing that every minimal $\hucw$ class interprets grids would suffice to prove Seese's conjecture.  Indeed, if $\cl{C}$ interprets grids of unbounded size, so does every class that contains $\cl{C}$.
However, Korpelainen has shown~\cite{Korpelainen16} that there are $\hucw$ classes that contain no minimal $\hucw$ subclass.  We give a construction of such a class  in Section~\ref{sec:non-minimal}.
This is linked to the existence of antichains of unbounded clique-width.  Specifically, we establish the following facts.
\begin{enumerate}
\item If $\cl{C}$ is a minimal $\hucw$ class, then it cannot contain an antichain of unbounded clique-width (Theorem~\ref{thm:minimal-antichain} in Section~\ref{sec:antichain}).
\item If $\cl{C}$ is an $\hucw$ class that contains no minimal class, it must contain an antichain of unbounded clique-width (Theorem~\ref{thm:non-minimal-antichain} in Section~\ref{sec:antichain}).
\end{enumerate}

From these, the theorem below follows, which suggests a programme for proving Seese's conjecture.
\begin{theorem}\label{thm:strong-seese}
  The strong Seese conjecture holds if, and only if, both of the following are true:
  \begin{enumerate}
  \item every antichain of unbounded clique-width interprets grids; and 
  \item every minimal $\hucw$ class interprets grids.
  \end{enumerate}
\end{theorem}

\subsection{Antichains and Minimal Classes}\label{sec:antichain}
We first establish the relationship between the existence of antichains of unbounded clique-width and the minimality of $\hucw$ classes.  These are established in Theorems~\ref{thm:minimal-antichain} and~\ref{thm:non-minimal-antichain}.

We say that a sequence $(\cl{C}_i)_{i \in \omega}$ is an infinite
\emph{strictly descending $\hucw$-chain} if for each $i$,  $\cl{C}_i$ is an $\hucw$ class and $\cl{C}_{i+1}$ is a proper subclass of $\cl{C}_i$.  We say that $\cl{C}$ contains an infinite strictly descending $\hucw$-chain if there is such a chain with $\cl{C}_i \subseteq \cl{C}$ for all $i$.

\begin{lemma}\label{lemma:helper-main}
  The following are equivalent:
  \begin{enumerate}[nosep]
    \item \vspace{2pt} $\cl{C}$ contains an infinite strictly
      descending $\hucw$-chain whose intersection is a class of
      bounded clique-width.\label{helper-main:1}
    \item \vspace{2pt} $\cl{C}$ contains an infinite strictly
      descending $\hucw$-chain whose intersection is empty.\label{helper-main:2}
    \item \vspace{2pt} $\cl{C}$ contains an antichain of
      unbounded clique-width.\label{helper-main:3}
  \end{enumerate}
\end{lemma}

\begin{proof}
  (\ref{helper-main:3}) $\rightarrow$ (\ref{helper-main:2}): If
  $\{G_1, G_2, \ldots\}$ is such an antichain, then let $\cl{C}_i$ be
  the hereditary closure of $\{G_i, G_{i+1}, \ldots\}$ for $i \ge 1$. Then $\cl{C}_1
  \supsetneq \cl{C}_2 \supsetneq \ldots$ is an infinite strictly
  descending $\hucw$-chain whose intersection is empty.

  \vspace{2pt}(\ref{helper-main:2}) $\rightarrow$ (\ref{helper-main:1}): Trivial since the empty class has clique-width 0. 
  
  \vspace{2pt}(\ref{helper-main:1}) $\rightarrow$
  (\ref{helper-main:3}): Let $\cl{C}_1 \supsetneq \cl{C}_2 \supsetneq
  \ldots$ be such a descending $\hucw$-chain and $\cl{C}_\omega =
  \bigcap_{i \ge 1} \cl{C}_i$. Let $\cl{D}_i = \cl{C}_i \setminus
  \cl{C}_{i+1}$ for $i \ge 1$. Then for $1 \leq i < j$, we have
  $\cl{D}_i \cap \cl{C}_j = \emptyset$; hence $\cl{D}_i \cap \cl{D}_j
  = \cl{D}_i \cap \cl{C}_\omega = \emptyset$. Further, $\cl{C}_i =
  \big(\bigcupdot_{i \leq k < \omega} \cl{D}_k \big) \bigcupdot
  \cl{C}_{\omega}$. 

  \begin{claim}\label{claim:helper}
    The following are true:
    \begin{enumerate}[nosep]
      \item \vspace{2pt} For $1 \leq i < j$, no graph in $\cl{D}_i$ is
        an induced subgraph of a graph in
        $\cl{D}_j$.\label{helper-claim:1}
      \item \vspace{2pt} For $i \ge 1$, for every graph $G \in
        \cl{D}_i$, there exists a number $f(G) > i$ such that for all $j \ge f(G)$, no graph
        in $\cl{C}_{j} \setminus \cl{C}_\omega$ is an induced
        subgraph of $G$.\label{helper-claim:2}
    \end{enumerate}
  \end{claim}
  \begin{proof}
    (\ref{helper-claim:1}): If $G \subseteq H$ for some $G \in
    \cl{D}_i$ and $H \in \cl{D}_j$, then since $\cl{D}_j \subseteq
    \cl{C}_j$ and $\cl{C}_j$ is hereditary, we would have $G \in
    \cl{C}_j$; but that contradicts the fact that $\cl{D}_i \cap
    \cl{C}_j = \emptyset$.

    (\ref{helper-claim:2}): Let $H_1, \ldots, H_r$ be an enumeration
    of the induced subgraphs of $G$ that are not in $\cl{C}_\omega$ --
    clearly $r$ is finite since $G$ is finite. Since $\cl{C}_i =
    \big(\bigcupdot_{i \leq j < \omega} \cl{D}_j \big) \bigcupdot
    \cl{C}_\omega$, there exist numbers $j_1, \ldots, j_r \in [i,
      \omega)$ such that $H_i \in \cl{D}_{j_i}$ for $i \in \{1,
      \ldots, r\}$. It then follows by the properties of the
      $\cl{D}_i$'s above that $f(G) = \max\{j_i \mid 1 \leq i \leq r\}
      + 1$ is indeed as desired.
  \end{proof}
  
  We now use the above claim to inductively construct an antichain
  of $\cl{C}$ of unbounded clique-width. Let $G_0$ be a graph in $\cl{D}_0$. Assume that we have constructed graphs $G_0, \ldots, G_{i}$ for $i \ge 0$ such that (i) $G_j \in \cl{D}_{l_j}$ and $l_j > l_{j - 1}$ for $1 \leq j \leq i$; (ii) $\{G_0, \ldots, G_{i}\}$ is an antichain; and (iii) the clique-width of $G_j$ is strictly greater than that of $G_{j-1}$ for $1 \leq j \leq i$. Let $k = \max \{f(G_j) \mid 1 \leq j \leq i\} > l_i$ where $f$ is as in Claim~\ref{claim:helper}. Consider the
  class $\cl{C}_k \setminus \cl{C}_\omega$ -- by
  Lemma~\ref{claim:helper}, all graphs in this class are incomparable
  with each of $G_0, \ldots, G_i$ in the induced subgraph order. Further, since $\cl{C}_k$
  has unbounded clique-width while $\cl{C}_\omega$ has bounded clique
  width, we have that $\cl{C}_k \setminus \cl{C}_\omega$ has unbounded
  clique-width, whereby there exists $G_{i+1} \in \cl{C}_k \setminus
  \cl{C}_\omega$ such that $G_{i+1}$ has clique-width greater than that of
  $G_i$. Let $l_{i+1} \ge k > l_i$ be such that  $G_{i+1} \in \cl{D}_{l_{i+1}}$. Then we see that $G_{i+1}$ is indeed as desired to complete the induction.
\end{proof}

We are now ready to prove the two results linking minimality of $\hucw$ classes and the existence of antichains of unbounded clique-width.
\begin{theorem}\label{thm:minimal-antichain} 
  If $\cl{C}$ is a minimal $\hucw$ class, then $\cl{C}$ does not contain an antichain of unbounded clique-width.
\end{theorem}
\begin{proof}
  If $\cl{C}$ contains an antichain of unbounded clique-width, then by Lemma~\ref{lemma:helper-main}, we have that $\cl{C}$ contains an infinite strictly descending $\hucw$-chain, and hence in particular a proper subclass that is $\hucw$.  Hence $\cl{C}$ is not minimal.
\end{proof}

\begin{theorem}\label{thm:non-minimal-antichain}
  If $\cl{C}$ is $\hucw$ and does not contain a minimal $\hucw$ class, then there exists in $\cl{C}$ an antichain of unbounded clique-width.
\end{theorem}
\begin{proof}
  We assume without loss of generality that the vertices of the graphs of~$\cl{C}$ belong to the set $\mathbb{N}$ of natural numbers, so that $\cl{C}$ is countable. Suppose that~$\cl{C}$ does not contain a minimal class.  Consider the sequence $(\cl{C}_\lambda)_{\lambda \ge 0}$ of classes of structures, for ordinals $\lambda$, defined inductively as follows.  Let $\cl{C}_0 = \cl{C}$ and inductively, assume that for all $\nu < \lambda$, the class $\cl{C}_\nu$ has been defined and that $\cl{C}_\nu \subseteq \cl{C}$ for all $\nu < \lambda$. If $\lambda$ is a limit ordinal, define $\cl{C}_\lambda = \bigcap_{\nu < \lambda} \cl{C}_\nu$. If $\lambda$ is a successor ordinal of say $\lambda^-$, then define $\cl{C}_\lambda$ as follows. If $\cl{C}_{\lambda^-}$ is not $\hucw$, then $\cl{C}_{\lambda} = \cl{C}_{\lambda^-}$.  Otherwise $\cl{C}_{\lambda^-}$ is $\hucw$ and $\cl{C}_{\lambda^-} \subseteq$ $\cl{C}$; then $\cl{C}_{\lambda^-}$ cannot be minimal since by our premise, $\cl{C}$ does not contain any minimal $\hucw$ class.  Let $\cl{C}_\lambda$ be any proper subclass of $\cl{C}_{\lambda^-}$ that is $\hucw$. This completes the construction of the sequence $(\cl{C}_\lambda)_{\lambda \ge 0}$.

  Consider now the set $\mc{P}$ of ordinals defined as $\mc{P} = \{\lambda \mid \cl{C}_\lambda \text{ is not } \hucw\}$. This set is non-empty -- since $\cl{C}$ is a class of finite graphs whose vertices are natural numbers, $\cl{C}$ is countable and hence $\cl{C}_\lambda = \emptyset$ for all uncountable $\lambda$. By the definition above, if $\lambda \in \mc{P}$, then all ordinals greater than $\lambda$ are in $\mc{P}$ as well. Now since the ordinals are well ordered, $\mc{P}$ has a minimum, call it $\mu^*$. We make the following observations about $\mu^*$:
  \begin{enumerate}
    \item ${\mu^*}$ must be a limit ordinal. If it is a successor
      ordinal of say $\lambda$, then~$\cl{C}_\lambda$ must be 
      $\hucw$ since ${\mu^*}$ is the minimum ordinal in $\mc{P}$. But
      if $\cl{C}_\lambda$ is  $\hucw$, then $\cl{C}_{\mu^*}$ must
      be a $\hucw$ class by the inductive definitions
      above. Therefore, $\cl{C}_{\mu^*} = \bigcap_{\nu < {\mu^*}}
      \cl{C}_\nu$ where $\cl{C}_\nu$ is $\hucw$ for all $\nu < \mu^*$.\label{obs:1}
    \item $\mu^*$ is countable -- this is because $\cl{C}$ is countable.
    \item $\cl{C}_{\mu^*}$ is a hereditary class of bounded clique-width. Let $G \in \cl{C}_{\mu^*}$ and $H \subseteq G$. Then by
      (\ref{obs:1}) above, $G \in \cl{C}_\nu$ for all $\nu <
      {\mu^*}$. Since each $\cl{C}_\nu$ is hereditary, we have $H \in \cl{C}_\nu$ for all
      $\nu < {\mu^*}$. Then $H \in \cl{C}_{\mu^*}$. So
      $\cl{C}_{\mu^*}$ is hereditary. That~$\cl{C}_{\mu^*}$ has
      bounded clique-width now follows from the fact that
      $\cl{C}_{\mu^*}$ is not $\hucw$.
  \end{enumerate}

  Now since $\mu^*$ is countable, it has cofinality $\omega$ so that there exists an increasing function $f: \mathbb{N}  \rightarrow \mu^*$ (where $\mu^*$ is seen as the set of ordinals less than $\mu^*$) such that if $\cl{F}_i = \cl{C}_{f(i)}$ for  $i \in \mathbb{N}$, then $\bigcap_{i \in \mathbb{N}} \cl{F}_i = \cl{C}_{\mu^*}$. We observe that  
  $\cl{F}_1 \supsetneq \cl{F}_2 \supsetneq \ldots$ is an infinite
  strictly descending $\hucw$-chain in $\cl{C}$, whose intersection $\cl{C}_{\mu^*}$ is a class of bounded
  clique-width. It now follows by Lemma~\ref{lemma:helper-main} that
  $\cl{C}$ contains an antichain of unbounded clique-width.
\end{proof}

The converse of Theorem~\ref{thm:non-minimal-antichain} does not hold.  That is to say, we can construct an $\hucw$ class that both contains a minimal $\hucw$ class and contains an antichain of unbounded clique-width.  Indeed, if $\mc{C}_1$ is a minimal $\hucw$ class and $\mc{C}_2$ the hereditary closure of an antichain of unbounded clique-width then clearly  $\mc{C} = \mc{C}_1 \cup \mc{C}_2$ has this property.

\subsection{HUCW Classes which Contain No Minimal Class}\label{sec:non-minimal}

Theorem~\ref{thm:non-minimal-antichain} raises the obvious question of whether there exists any class~$\cl{C}$ which is $\hucw$ but does not contain a minimal $\hucw$ class.  The existence of such a class was demonstrated by Korpelainen~\cite{Korpelainen16}.  Here we give a similar construction which we arrived at independently.

 \begin{theorem}\label{thm:non-minimal-existence}
  There is an $\hucw$ class $\mathcal{T}$  that does not contain any minimal $\hucw$ class.
\end{theorem}

It suffices to show that
 if $\mathcal{C}$ is any hereditary subclass of $\mathcal{T}$ of unbounded
 clique-width, it contains an antichain  of unbounded clique-width.

 Towards this, let $G_{n,n}$ denote the $n \times n$ grid.
Note that, in $G_{n,n}$, every vertex has degree $2$, $3$ or $4$, and
there are exactly four vertices (at the corners) of degree~$2$.  For
$n \geq 3$, we define  $T_{n}$ as the graph obtained from
$G_{n,n}$ by:
\begin{enumerate}
\item removing every vertex $v$ of degree $2$ and inserting an edge
  between the two neighbours of $v$; and
\item replacing every vertex $v$ of degree $4$ by four new vertices
  $v_1, v_2, v_3, v_4$ that are connected in a $4$-cycle so that the
  four edges incident on $v$ are now each incident on one of the four
  new vertices.
\end{enumerate}
It is easily seen that $T_{n}$ is $3$-regular, and it is more
convenient to work with than grids.  The number of vertices in $T_n$ is less than $4n^2$. 

Recall that a graph $H$ is a \emph{subdivision} of a graph $G$ if it
is obtained from $G$ by replacing every edge by a simple path.  For a
positive integer $t$, we write $G^t$ for the $t$-subdivision of $G$: the graph
obtained from $G$ by replacing each edge of $G$ by a path of
length $t$.  We make the following simple observation for later use:
\begin{lemma}\label{lem:subdivision}
  If $H$ is a
subdivision of $G$ and $\tw(G)=k$, then $k \leq \tw(H) \leq \max(k,3)$.
\end{lemma}
\begin{proof}
  The lower bound on $\tw(H)$ follows immediately from the fact that $G$ is a minor of $H$ so $\tw(G) \leq \tw(H)$.
  
  Suppose now that $(T,\beta)$ is a tree decomposition of $G$ of width $k$.  To obtain a tree decomposition of $H$, consider an edge $\{u,v\}$ of $G$ which is subdivided into a path $u=p_0,\ldots, p_t = v$ in $H$.  As $\{u,v\}$ is an edge of $G$, there must be a node $a$ of $T$ such that $\{u,v\} \subseteq \beta(t)$.  We attach a path $a_1,\ldots,a_t$ of length $t$ to $a$ and let $\beta(a_i) = \{u,v,p_i,p_{i+1}\}$.  Doing this for each edge gives us a tree decomposition of $H$ whose width is  $\max(k,3)$.
\end{proof}

Define the class $\mathcal{T} = \{ H \mid H \subseteq T^n_n \text{ for some } n > 2\}$, i.e.\ the collection of graphs that are induced subgraphs of the $n$-subdivision of $T_n$ for some $n$.  We consider the graphs $H \in \mathcal{T}$ where every vertex has degree $2$ or $3$.  We call such graphs \emph{skeleton graphs} and  the
vertices of degree $3$  the \emph{branch vertices}.  Note that every graph in $\mathcal{T}$ is an induced subgraph of a skeleton graph.  

The next two lemmas establish some useful properties of the graphs in $\mathcal{T}$.
\begin{lemma}\label{lem:branch}
  If $H \in \mathcal{T}$ is a skeleton graph with at most $m > 2$ branch vertices, then $\cw(H) \leq 6m -2$.
\end{lemma}

\begin{proof}
  Since $H$ has at most $m$ branch vertices, it is the subdivision of some graph $G$ with $m$ vertices.    Hence, by Lemma~\ref{lem:subdivision}, the treewidth of $H$ is at most~$m$. 
  Note further that all graphs in $\mathcal{T}$ are planar and hence $H$ is planar.   For any planar graph $H$,  $\cw(H) \leq  6\tw(H)-2$~\cite[Thm~17]{courcelle2018}, and the result follows.
\end{proof}

\begin{lemma}\label{lem:unbounded}
 If $H$ is a subdivision of $T_n$ for $n > 2$, then the clique-width of $H$ is at least $(n-1)/6$.
\end{lemma}

\begin{proof}
Since $G_{n-2,n-2}$ is a minor of $T_n$ and $\tw(G_{k,k} = k$ we have that $\tw(T_n) \geq n-2$.  Also, by Lemma~\ref{lem:subdivision} we know that $\tw(H) = \tw(T_n)$.  Now, for any planar graph $G$ we have $\tw(G) \leq 6\cw(G) - 1$ by~\cite[Prop.~2.115]{courcelle-engelfriet}.  Since $H$ is planar, the result follows.
\end{proof}

\begin{proof}[Proof of Theorem~\ref{thm:non-minimal-existence}]
  The class $\mathcal{T}$ is hereditary by definition and has unbounded clique-width by Lemma~\ref{lem:unbounded}.  Thus, it remains to show that for every class $\mathcal{C} \subseteq \mathcal{T}$, if $\mathcal{C}$ has unbounded clique-width, then $\mathcal{C}$ contains an antichain of unbounded clique-width.

  So, suppose $\mathcal{C} \subseteq \mathcal{T}$ has unbounded clique-width. 
  For a graph $H \in \mathcal{T}$, write $\text{mn}(H)$ for the length of the shortest path between two branch vertices of $H$.  We define the following sequence of graphs.  First, let $G_0$ be any graph in $\mathcal{C}$ containing at least two branch vertices.  Suppose we have defined $G_i$ for $i \geq 0$, and let $t = \max(\cw(G_i),\text{mn}(G_i))$.  We then choose $G_{i+1}$ to be any graph in $\mathcal{C}$ with $\cw(G_{i+1}) > 24t^2 -2$.

  It is clear that the sequence of graphs $(G_i : i \in \omega)$ is of unbounded clique-width, since $\cw(G_{i}) < \cw(G_{i+1})$ for all $i$.  We now argue that this is also an antichain.  For any $i < j$, clearly $G_j$ cannot be an induced subgraph of $G_i$ since  $\cw(G_{i}) < \cw(G_j)$, so it remains to show that $G_i$ is not an induced subgraph of $G_j$.  Since $\cw(G_{j}) > 24t^2 -2$, where $t = \max(\cw(G_i),\text{mn}(G_i))$, it follows by Lemma~\ref{lem:branch} that $G_j$ has more than $4t^2$ branch vertices.  Since $T^n_n$ contains fewer than $4n^2$ branch vertices, it follows that $G_j$ is not an induced subgraph of $T^n_n$ for any $n  \leq t$.  Hence, $\text{mn}(G_j)$ is at least $t+1$.  However, by the choice of~$t$, $\text{mn}(G_i) \leq t$ and so $G_i$ contains two branch vertices at distance at most $t$.  We conclude that $G_i$ is not an induced subgraph of $G_j$.
\end{proof}

\section{Grid-Like Classes}\label{sec:grid-classes}
We begin our systematic exploration of all known minimal hereditary classes of unbounded clique-width.  Many such classes are defined in terms of a grid-like structure and this is used to show that they have unbounded clique-width.  The challenge in these cases is to show how this grid structure can be drawn out through an $\mso$ interpretation.  We begin with a collection of minimal $\hucw$ classes (indeed, an uncountable collection of them) defined in terms of certain infinite words and show in Section~\ref{section:word-defined-classes} that they interpret grids.  This is then extended by reductions in  Section~\ref{section:reductions} to a number of other classes.  

\subsection{Word-defined minimal classes}\label{section:word-defined-classes}

Our starting point is a construction given by Brignall and Cocks~\cite{BC22} to demonstrate that there are uncountably many minimal $\hucw$ classes,  extending a construction by Collins et al.~\cite{CollinsFKLZ18} showing the existence of infinitely many such classes.  They construct a hereditary class $\cl{S}_\alpha$ of graphs for each $\omega$-word $\alpha \in \{0, 1, 2, 3\}^{\omega}$ and show that as long as $\alpha$ contains infinitely many non-zero letters, the class $\cl{S}_\alpha$ has unbounded clique-width.   Moreover,  for uncountably  many distinct such $\alpha$, $\cl{S}_\alpha$ is also minimal.  The conditions under which  $\cl{S}_\alpha$ is minimal need not concern us here.  We are able to show that whenever $\alpha$ contains infinitely many non-zero letters  $\cl{S}_{\alpha}$ interprets grids via $\mso$ interpretations.  In particular, this covers all minimal classes $\cl{S}_{\alpha}$ of unbounded clique-width, including those defined in~\cite{CollinsFKLZ18}.  Before we proceed to a proof, we give a precise definition of the classes $\cl{S}_{\alpha}$.

The class $\cl{S}_\alpha$ is defined as the class of all finite induced subgraphs of a single countably infinite graph $\mc{P}_\alpha$.  The set of vertices of $\mc{P}_\alpha$ is $\{v_{i, j} \mid i, j \in \mathbb{N} \}$.  We think of the set as an infinite collection of \emph{columns} $V_j = \{v_{i,j} \mid i \in \mathbb{N} \}$.  All edges are between vertices in adjacent columns, i.e.\ there is no edge between $v_{i,j}$ and~$v_{i',j'}$ unless $j' = j+1$ or $j'=j-1$.  The edges between successive columns are defined by the word $\alpha$ according to the following rules.
\begin{enumerate}
    \item  If $\alpha_j = 0$, then  $\{v_{i, j}, v_{k, j+1}\} \in E(\mc{P}_\alpha)$  if, and only if, $i=k$.
    \item If $\alpha_j = 1$, then $\{v_{i, j}, v_{k, j+1}\} \in E(\mc{P}_\alpha)$ if, and only if, $i \neq k$ for $i, k \in \mathbb{N}$.
    \item If $\alpha_j = 2$, then $\{v_{i, j}, v_{k, j+1}\} \in E(\mc{P}_\alpha)$ if, and only if, $i \leq k$ for $i, k \in \mathbb{N}$.
    \item If $\alpha_j = 3$, then $\{v_{i, j}, v_{k, j+1}\} \in E(\mc{P}_\alpha)$ if, and only if, $i \ge k$ for $i, k \in \mathbb{N}$.
     \end{enumerate}
 The class $\cl{S}_\alpha$ is now given by $\cl{S}_\alpha = \{G \mid G~\mbox{is a finite induced subgraph of}~\mc{P}_\alpha\}$.  We show the following theorem in this section.

\begin{theorem}\label{theorem:word-defined-classes}
Let $\alpha \in \{0, 1, 2, 3\}^{\omega}$ be such that $\alpha$ contains infinitely many non-zero letters. Then there exists an $\mso$ interpretation $\Theta$ such that $\Theta(\cl{S}_\alpha)$ contains the class of all square grids.
\end{theorem}

To prove Theorem~\ref{theorem:word-defined-classes}, we show the existence of an $\mso$ interpretation $\Psi$ such that the hereditary closure of $\Psi(\cl{S}_\alpha)$ contains the class of all square grids.  Proposition~\ref{prop:reduction-helper} ensures that this indeed suffices.  It is clear that graphs in $\cl{S}_{\alpha}$ have a built-in grid-like structure with vertices arranged in rows and columns.  The main challenge is to show that a sufficient part of this structure can be made explicit using an $\mso$ interpretation.  We give an outline of the construction.

What we show is that we can find in $S_{\alpha}$ a sequence of graphs $G_n$ for $n \in \mathbb{N}$ within which we can interpret \emph{upper triangular grids}.  One can think of an upper triangular grid $U_{t}$ as the subgraph of the $t \times t$ grid induced by the vertices above the main diagonal, i.e.\ those vertices in the set $\{(i,j) \mid 1\leq i,j \leq t\}$ with $i \leq j$.  It is clear that $U_t$ has as an induced subgraph an $r \times r$ grid, where $r = \lfloor \frac{t}{2} \rfloor$.

Let $\alpha \in \{0, 1, 2, 3\}^\omega$ be an $\omega$-word containing infinitely many non-zero letters.  We write $\alpha_i$ for the $i^{\text{th}}$ letter of  $\alpha$.
Let $p < \omega$ be the least value such that $\alpha_p \neq 0$.  Fix $n \ge 1$ and let $l$ be the length of the shortest contiguous subsequence of $\alpha$ starting at $\alpha_p$ that contains exactly $2n+2$ elements which are not $0$.  We write $\beta_0\cdots\beta_{l-1}$ for this sequence, so $\beta_0 = \alpha_p$.

Recall that the vertices of $P_{\alpha}$ are $\{v_{i, j} \mid i, j \in \mathbb{N} \}$, and we write $V_j$ for the set $\{v_{i,j} \mid i \in \mathbb{N} \}$.  We define the graph $G_n$ to be the subgraph of $P_{\alpha}$ induced by the set $\col{-1} = \bigcup\limits_{i = 0}^{i = l-1} \col{i}$ where $\col{i} \subseteq V_{p+i}$ is defined as follows for $0 \leq i < l$.
\begin{enumerate}
    \item $\col{0} = \{v_{0,p}, v_{1,p},v_{2,p}, v_{3, p}\}$; and 
    \item $\col{i+1} =  \{v_{0,p+i+1},v_{1,p+i+1},\ldots,v_{t-1,p+i+1}\}$ where $t = |\col{i}|$ if $\beta_{i+1} = 0$ and $t = |\col{i}|+1$ otherwise.
\end{enumerate}

\begin{figure}[ht]
    \centering
    \includegraphics[scale=0.7]{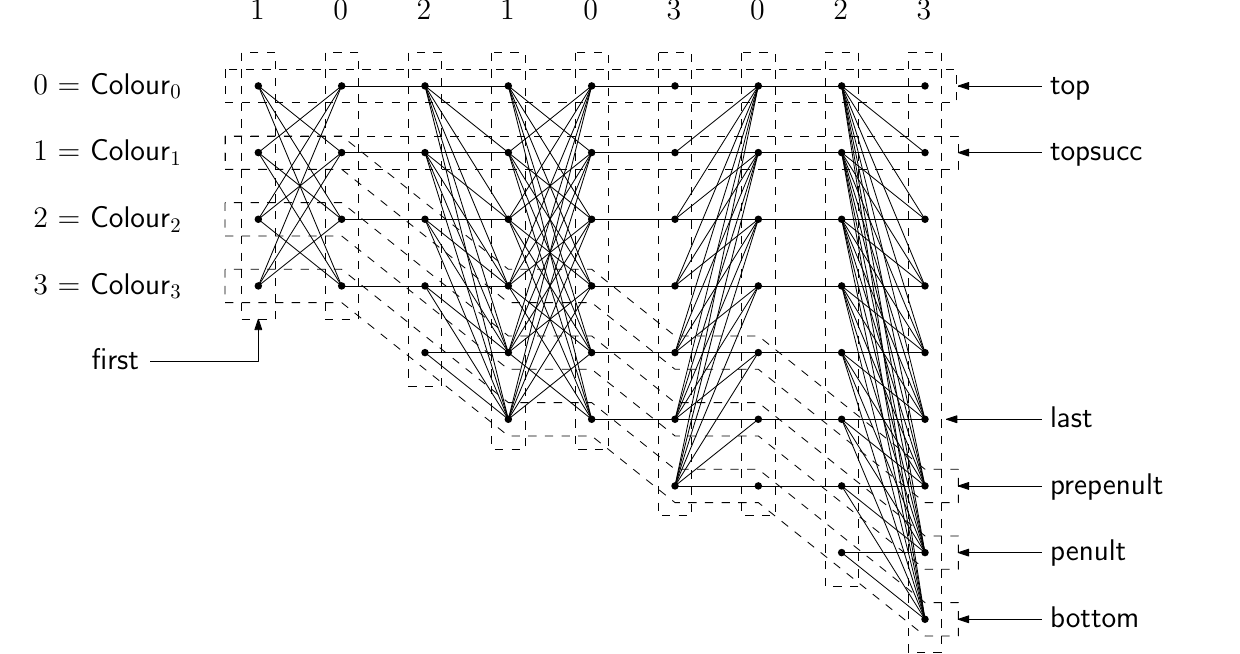}
    \caption{The graph $H_2$ for $\alpha = (102103023)^\omega$.  The unlabeled graph underlying $H_2$ is $G_2$.}
    \label{fig:word-defined-classes-Hn}
\end{figure}

It is clear that $G_n \in \cl{S}_\alpha$.  We show that we can interpret upper triangular grids in this class of graphs.  The key challenge in defining the required interpretation is to define the two binary relations: one that relates vertices that are in the same  column and the other that relates vertices that are in the same row.  In constructing the interpretation we make use of a number of set parameters to obtain a labeled version $H_n$ of $G_n$; in particular, $H_n$ uses unary predicates for the vertices corresponding to the possible values of $\beta_i$, for the first and last column, the top, bottom and penultimate rows, and the rows immediately succeeding and preceding the top and penultimate rows respectively.  The ``diagonal'' nature of the bottom row is vital to allowing us to define when two vertices are in successive columns, which we need in order to define the two relations of being in the same row and in the same column. We now give all the technical details below.

We define the graph $H_n$ as the expansion of $G_n$ with unary predicates $\mathsf{Colour}_{0}$, $\mathsf{Colour}_{1}$, $\mathsf{Colour}_{2}$, $\mathsf{Colour}_{3}$, $\mathsf{top}$, $\mathsf{topsucc}$, $\mathsf{bottom}$, $\mathsf{penult}$, $\mathsf{prepenult}$, $\mathsf{first}$ and $\mathsf{last}$ which are interpreted as follows.  For $i \in \{0, 1, 2, 3\}$, the predicate $\mathsf{Colour}_{i}$ is interpreted as the set $\bigcup\limits_{\beta(j) = i} \col{j}$; $\mathsf{top}$ is interpreted as the top row of $G_n$; $\mathsf{topsucc}$ is interpreted as the second row of $G_n$ after the top;  $\mathsf{bottom}$ is the set of all vertices $v_{i,j}$ such that
$i' \leq i$ for all $v_{i',j} \in \col{j}$; $\mathsf{penult}$ is the set of all vertices $v_{i,j}$ such that $v_{i+1,j}$ is in $\mathsf{bottom}$; $\mathsf{prepenult}$ is the set of all vertices $v_{i,j}$ such that $v_{i+2,j}$ is in $\mathsf{bottom}$; and finally, $\mathsf{first}$ and $\mathsf{last}$ are interpreted as the sets $\col{0}$ and $\col{l-1}$ respectively.  Figure~\ref{fig:word-defined-classes-Hn} provides an illustration.

We now describe the construction of the interpretation $\Psi$. Towards this, we need a number of auxiliary predicates which we define below.


\begin{enumerate}
    \item $\samecol{x, y}$: This predicate is true of $x, y$ in $H_n$ if $x, y \in \col{i}$ for some $i \in \{0, \ldots, l-2\}$ and $\beta_i \in \{1, 2, 3\}$ as long as neither of $x, y$ is in  $\mathsf{top}$ or $\mathsf{bottom}$. 
    \[
    \begin{array}{rll}
        \rlboundary{x} & := & \top{x} \vee \bottom{x} \\
        \samecol{x, y} & := & \neg (\rlboundary{x} \vee \rlboundary{y}) \wedge\\
        & & \neg \mycolor{0}{x} \wedge \bigwedge\limits_{i = 1}^{i = 3} \mycolor{i}{x} \leftrightarrow \mycolor{i}{y} \wedge\\
        & & \forall z \big( (\neg \mycolor{3}{x} \wedge \bottom{z}) \bigvee\\
        & & ~~~~~(\mycolor{3}{x} \wedge \top{z})\big)\\
        & & ~~~~~\rightarrow (E(x, z) \leftrightarrow E(y, z))
    \end{array}
    \]

    To understand the last condition, note that if $x$ and $y$ are in the same column  $\col{i}$ that is not the last, with $\beta_i \in \{1,2,3\}$, and neither of $x$ or $y$ is in the top or bottom row in $\col{i}$, then the bottom elements of $\col{i-1}$ and $\col{i+1}$ are neighbours of either both $x$ and $y$ or neither; likewise for the top elements of $\col{i-1}$ and~$\col{i+1}$.  On the other hand,   suppose $x$ and $y$ are in different columns, say $\col{i}$ and $\col{j}$ respectively with $i < j$.  Since $\beta_j$ is 1, 2 or 3, we know  that $|\col{j+1}| = |\col{j}| + 1$, and hence if $\beta_j$ is 1 or 2, then every element of $\col{j}$,  in particular $y$, is adjacent to the bottom element $z$ of $\col{j+1}$, and if $\beta_j$ is 3, then every element of $\col{j}$, and $y$ in particular, is adjacent to the top element $z'$ of $\col{j+1}$.  Since $x$ is not in a column adjacent to $\col{j+1}$, it cannot have an edge to either $z$ or $z'$, and hence $x$ and $y$ do not satisfy the predicate $\mathsf{samecolumn}$.

    \item $\adjcol{x, y}$: This predicate is true of $x, y$ in $H_n$ if for some $i, j$ with  $|i - j| = 1$ and $\beta_i \in \{1, 2, 3\}$, it holds that $x$ is in $\col{i}$ but not in  $\mathsf{top}$ or $\mathsf{bottom}$, and $y \in \col{j}$.
    \[
    \begin{array}{lll}
    \adjcol{x, y} & := & \neg (\rlboundary{x} \vee \mycolor{0}{x}) ~\wedge \\
    & &\exists u\, (\samecol{x, u} \wedge E(u, y))\\
    \end{array}
    \]
    \item $\domain{x}$; This predicate is true of $x$ in $H_n$ if it is not one of the ``periphery'' vertices of $G_n$.
    \[
    \domain{x} := \neg (\rlboundary{x} \vee \topsucc{x} \vee \penult{x} \vee \first{x} \vee \last{x})
    \]
   \item \noindent $\rhscol{i; S}{x, y}$: For $i \in \{0, 1, 2, 3\}$ and $S \subseteq \{0, 1, 2, 3\}$, this predicate is true of $x, y$ in $H_n$ if $\domain{x}$ and $\domain{y}$ both hold and if  $x \in \col{j}$ and $y \in \col{j+1}$ for some $j$ with $\beta_j = i$ and $\beta_{j+1} \in S$. 
    We need this predicate only for the following specific values of $[i; S]$: (i) $[0; \{1, 2, 3\}]$, (ii) $[1; \{0, 2, 3\}]$, (iii) $[2; \{0, 1, 2, 3\}]$, and (iv) $[3; \{0, 1, 2, 3\}]$.  
    
    \[
    \begin{array}{rll}
	    \rhscol{3; \{0, 1, 2\}}{x, y} & := & \domain{x} \wedge \mycolor{3}{x} \wedge \domain{y} \wedge \\
        & & \bigwedge\limits_{i = 0}^{i = 3} \mycolor{i}{y} \rightarrow \eta_i(x, y)\\
        \eta_0(x, y) & := & \adjcol{x, y} \wedge \exists v\, (\samecol{x, v}  \\
		& & \wedge \penult{v} \wedge E(v, y))\\
		\eta_1(x, y) & := & \adjcol{x, y} \wedge\\
        & & \exists v\, \big(\samecol{x, v} \wedge \prepenult{v} \wedge \\
        & & ~~~~\exists w (\samecol{y, w} \wedge \neg E(w, v))\big)
    \end{array}
    \]
    \[
    \begin{array}{rll}
        \eta_2(x, y) & := & \eta_1(x, y)\\[3pt]
        \eta_3(x, y) & := & \adjcol{x, y} \wedge\\
        & & \exists v\, \big(\samecol{x, v} \wedge \prepenult{v} \wedge \\
        & & ~~~~\exists w (\samecol{y, w} \wedge E(w, v))\big)\\[3pt]
        \rhscol{2; \{0, 1, 2\}}{x, y} & := & \domain{x} \wedge \mycolor{2}{x} \wedge \domain{y} \wedge \\
        & & \bigwedge\limits_{i = 0}^{i = 3} \mycolor{i}{y} \rightarrow \eta_i(x, y)\\[3pt]
        \rhscol{1; \{0, 2, 3\}}{x, y} & := & \domain{x} \wedge \mycolor{1}{x} \wedge \domain{y} \wedge \\
        & & \neg \mycolor{1}{y} \wedge \\
        & & \bigwedge\limits_{i \in \{0, 2, 3\}} \mycolor{i}{y} \rightarrow \eta_i(x, y)\\[3pt]

    
      \rhscol{0; \{1, 2, 3\}}{x, y} & := & \domain{x} \wedge \mycolor{0}{x} \wedge \domain{y} \wedge \\
        & & \neg \mycolor{0}{y} \wedge \adjcol{y, x} \wedge\\
        & & \forall v\, (\samecol{y, v} \wedge E(v, x)) \rightarrow\\
        & & \neg \big( (\mycolor{3}{v} \wedge \penult{v}) \vee\\
        & & ~~~   (\mycolor{2}{v} \wedge \topsucc{v})\big)
    \end{array}
  \]
These predicates are meant to give an orientiation to some edges in the symmetric relation $\adjcol{x, y}$.  Thus, it is sufficient to argue that if~$x$ is in $\col{i}$, then $y$ cannot be in $\col{i-1}$.  We present the argument for the case when $\beta_i = 3$. Other cases can be argued similarly.
    
Suppose $\beta_i = 3$.  There are four subcases depending on the value of $\beta_{i-1}$.  If $\beta_{i-1} = 0$, then the only element $z$ of $\col{i-1}$ that is adjacent to the penultimate element of $\col{i}$ is the bottom element of $\col{i-1}$.  But then $\domain{z}$ does not hold.  Thus no $y \in \col{i-1}$ satisfies the formula $\eta_0(x, y)$. If $\beta_{i-1} \in \{1, 2\}$, then the only element of $\col{i-1}$ that is not adjacent to the element of $\col{i}$ that satisfies $\mathsf{prepenult}$ is  the penultimate element of $\col{i-1}$ if  $\beta_{i-1}=1$ or the bottom element of $\col{i-1}$if $\beta_{i-1}=2$. But neither of these elements is in $\mathsf{domain}$. Finally, if $\beta_{i-1} =3$, then the only elements of $\col{i-1}$ that are adjacent to the element of $\col{i}$ that satisfies $\mathsf{prepenult}$ are  the penultimate and bottom elements of $\col{i-1}$; but again, neither of these elements is in $\mathsf{domain}$.

  \item \noindent $\hedge{x, y}$: This predicate is true of $x, y$ in $H_n$ if both  $\domain{x}$ and~$\domain{y}$ hold, $x$ and $y$ are in the same row and adjacent columns of~$H_n$ and either (i) $x \in \col{i}$ and $y\in \col{i+1}$ for some $i$; or (ii) $y \in \col{i}$ and $x\in \col{i+1}$ with $\beta_i = \beta_{i+1} \notin \{2, 3\}$.
    
    \[
    \begin{array}{rll}
        \hedge{x, y} & := & \domain{x} \wedge \domain{y} \wedge \bigwedge\limits_{i = 0}^{i = 3} \mycolor{i}{x} \rightarrow \gamma_i(x, y) \\[3pt]
        \gamma_3(x, y) & := &  \rhscol{3; \{0, 1, 2, 3\}}{x, y} \wedge E(x, y) \wedge\\
        & & \forall z (\rhscol{3, \{0, 1, 2, 3\}}{x, z} \wedge\\
        & & \lessthan{y, z, x} \wedge z \neq y) \rightarrow \neg E(x, z) \\[3pt]
        \lessthan{y, z, x} & := & \forall v (\samecol{x, v} \wedge E(z, v)) \rightarrow E(y, v)\\[3pt]
        \gamma_2(x, y) & := &  \rhscol{2; \{0, 1, 2, 3\}}{x, y} \wedge E(x, y) \wedge\\
        & & \forall z (\rhscol{2, \{0, 1, 2, 3\}}{x, z} \wedge\\
        & & \lessthan{y, z, x} \wedge z \neq y) \rightarrow \neg E(x, z) \\[3pt]
    \end{array}
    \]
    \[
    \begin{array}{rll}
        \gamma_1(x, y) & := & \big(\mycolor{1}{y} ~\wedge\\
        & & ~~(\exists z (\samecol{y, z} \wedge E(x, z))) \wedge \neg E(x, y)\big) \\
        & & \vee ~\big(\rhscol{1; \{0, 2, 3\}}{x, y} \wedge \neg E(x, y)\big)\\[3pt]
        \gamma_0(x, y) & := & \big(\mycolor{0}{y} \wedge E(x, y)\big) ~\vee \\
        & & \big(\rhscol{0; \{1, 2, 3\}}{x, y} \wedge E(x, y)\big)\\
    \end{array}
    \]
Suppose $x \in \col{i}$ and $y \in \col{j}$.  In all cases in the definition above except when $\beta_i = \beta_j \in \{0, 1\}$, it is the case that $\rhscol{\cdot}{x, y}$ is true, which means $j = i+1$.  In the case when $\beta_i = \beta_j =0$, we see that $x$ and $y$ are required to be adjacent, and when $\beta_i = \beta_j =1$, $x$ and $y$ are required to be non-adjacent with the additional condition that there is some element~$z$ in the same column as $y$ that is adjacent to $x$ -- both of these cases can happen only when  $x$ and $y$ are in adjacent columns and in the same row. We therefore are left with arguing that when $j = i+1$, then $x$ and $y$ satisfy~$\hedge{x, y}$ if, and only if, they are in the same row.

If $x \in \col{i}$ and $\beta_i = \{0,1\}$, the element $y$ of $\col{i+1}$ that is in the same row as~$x$ is easily distinguished.  If $\beta_i = 0$, $y$ is the only element of $\col{i+1}$ that is adjacent to $x$ and if $\beta_i = 1$ it is the only element of $\col{i+1}$ not adjacent to~$x$.  When $\beta_i=3$, we see that for elements $z$  and $y$ of  $\col{i+1}$ appearing in say  the rows $j$ and $j'$, we have $j \leq j'$ if, and only if, every element of $\col{i}$ (the column of $x$) that is adjacent to $z$ is also adjacent to $y$. This is expressed by the predicate $\lessthan{y, z, x}$.   With this linear order on $\col{i+1}$ defined, we see that an element $y$ of $\col{i+1}$ is in the same row as $x$ if, and only if, $x$ and $y$ are adjacent, and no element of $\col{i+1}$  that is less than $y$ is adjacent to $x$. Analogous arguments can be given for the final case of $\beta_i = 2$.

\item $\vedge{x, y}$: This predicate is true of $x, y$ in $H_n$ if  $\domain{x}$ and $\domain{y}$ both hold, and for some $i, j$, both $x$ and $y$ appear in the column $\col{j}$ such that $\beta_i \neq 0$, and $x$ appears in row $i$ and $y$ in row $i+1$.  In the following definition, $\mathsf{TC}\hedge{x,y}$ denotes that the pair $(x,y)$ is in the reflexive and transitive closure of $\mathsf{H}\text{-}\textsf{edge}$.  The reflexive and transitive closure of any binary relation is easily defined in $\mso$.
    \[
    \begin{array}{rll}
        \vedge{x, y} & := &\neg \mycolor{0}{x} \wedge \samecol{x, y}  \wedge\\
        & & \exists u \exists v \big(\prepenultedge{u, v} \wedge\\ 
        & & \hspace{0.9cm} \big(\mathsf{TC}\hedge{u,x} \wedge \mathsf{TC}\hedge{v,y} \big)\\
        \prepenultedge{u, v} & := & \neg \mycolor{0}{u} \wedge \samecol{u, v} \wedge \prepenult{v} ~\wedge \\
        & & \exists z (\prepenult{z} \wedge \hedge{z, u})\\
    \end{array}
  \]
  The formula $\prepenultedge{u, v}$ defines those pairs $(u,v)$ in the domain for which for some $i, j$, the vertices $u, v$ belong to $\col{i}$, and appear resp. in the rows $j$ and $j+1$, with the bottom element of $\col{i}$ appearing in row $j+3$.  To see why this definition is correct, note that $\prepenult{v}$ is true precisely when $v$ appears in row $j+1$ in $\col{i}$ (if the bottom element of $\col{i}$ appears in row $j+3$).  To identify $u$ in row $j$ of $\col{i}$, we exploit crucially the special way in which the columns were chosen in $G_n$: if $\beta_i\in\{1,2,3\}$ then $|\col{i}| = |\col{i-1}| + 1$.  This ensures that $\hedge{z,u}$ holds for the element $z \in \col{i-1}$ for which $\prepenult{z}$ holds.  Given the definition of $\prepenultedge{u, v}$ we see that every pair $(x, y)$ of elements  in the domain that are in the same column and in consecutive rows, is just a ``horizontal translate'' of a pair $(u, v)$ satisfying $\prepenultedge{u, v}$.  That is, $x$ and $y$ are reachable from $u$ and $v$ respectively by $\mathsf{H}\text{-}\textsf{edge}$-paths.

\end{enumerate}

We are now ready to define the $\mso$ interpretation $\Psi$.  Define an ``upper triangular'' $r \times r$ grid as the graph $U_r$ whose vertex set is $\{u_{i, j} \mid 1 \leq j \leq r, i \leq j\}$ and whose edge set is $\{\{u_{i, j}, u_{i, j+1}\}  \mid 1 \leq j < r, i \leq j\}\cup \{\{u_{i, j}, u_{i+1, j}\}  \mid 1< j \leq~r, i < j\}$.  A \emph{uniform subdivision} of $U_r$ is the graph obtained by choosing a subset $S \subseteq \{1, \ldots, r-1\}$ and for each $j \in S$ and each $i \leq j$, replacing the edge $\{u_{i, j}, u_{i, j+1}\}$ with a path on $k_j$ vertices for some $k_j \ge 2$.  It is easy to show that there exists a parameterless $\mso$ interpretation $\Gamma$ from graphs to graphs such that if $Z$ is  a uniform subdivision of $U_r$, then $\Gamma(Z)$ is  $U_r$. Observe that $U_{2r}$ contains the $r \times r$ grid as an induced subgraph.

We now define $\Psi$ as the composition given by $\Psi = \Gamma \circ \Delta$ where $\Delta = (\Delta_V(x),$ $\Delta_E(x, y))$ is as below. The formulae below contain the predicates $\mathsf{Colour}_{0}, \mathsf{Colour}_{1},$ $\mathsf{Colour}_{2}, \mathsf{Colour}_{3}, \mathsf{top}$,  $\mathsf{topsucc}$, $\mathsf{bottom}, \mathsf{penult},$ $\mathsf{prepenult}, \mathsf{first}$ and $\mathsf{last}$ which constitute the parameters of $\Psi$.
\[
\begin{array}{lll}
     \Delta_V(x) & := & \domain{x}  \\
     \Delta_E(x, y) & := & \hedge{x, y} \vee \hedge{y, x} \vee \vedge{x, y} \vee \vedge{y, x} 
\end{array}
\]
We observe that for the graph $H_n$ defined above, $\Delta(H_n)$ is indeed isomorphic to a uniform subdivision of $U_{2n}$.  Then $\Psi(H_n)$ is isomorphic to $U_{2n}$ and hence contains the $n \times n$ grid as an induced subgraph. 

\begin{proof}[Proof of Theorem~\ref{theorem:word-defined-classes}]
Given $\alpha \in \{0, 1, 2, 3\}^\omega$ containing infinitely many non-zero letters, consider the MSO interpretation $\Psi$ as described above. The hereditary closure of $\Psi(\cl{S}_\alpha)$ contains the class of all square grids. Taking $\cl{C}$ in Proposition~\ref{prop:reduction-helper} to be the class of all square grids, $\cl{D}$ to be $\cl{S}_\alpha$ and $\Xi$ to be $\Psi$, we are indeed done.
\end{proof}

\subsection{Composing Interpretations}\label{section:reductions}

We now consider the classes of graphs shown to be minimal $\hucw$ in~\cite{Lozin11, ABLS15}, and prove that these interpret grids using Theorem~\ref{theorem:word-defined-classes} above. Specifically, we show that for each class $\cl{C}$ among them, there is some $\alpha \in \{0, 1, 2, 3\}^\omega$ and an MSO interpretation $\Xi$ such that the hereditary closure of $\Xi(\cl{C})$  contains $\cl{S}_\alpha$.  Thus $\cl{C}$ interprets grids by Proposition~\ref{prop:reduction-helper}. 

\begin{theorem}\label{theorem:reductions}
The following minimal $\hucw$ classes of graphs interpret grids:
\begin{enumerate}
    \item Bichain graphs\label{bichain-graphs}
    \item Split permutation graphs\label{SP-graphs}
    \item Bipartite permutation graphs\label{BP-graphs}
    \item Unit interval graphs\label{UI-graphs}
\end{enumerate}
\end{theorem}

\begin{remark}
Note that Theorem~\ref{theorem:reductions}(\ref{UI-graphs}) follows from the results of Courcelle in~\cite{courcelle}. It is shown in~\cite{courcelle} that Seese's conjecture holds for the class of interval graphs. More specifically, it can be inferred from the results in~\cite{courcelle} that any unbounded clique-width subclass of interval graphs admits MSO interpretability of grids. It follows, in particular, that this is true of the unit interval graphs. We therefore show parts (\ref{bichain-graphs})--(\ref{BP-graphs}) of Theorem~\ref{theorem:reductions} to complete its proof.
\end{remark}

\noindent \textbf{Bichain graphs.}
\newcommand{\bichain}{\textsf{Bichain}}
We need some terminology to talk about these graphs. Given a graph $G$, a sequence $v_1,\ldots,v_k$ of vertices of $G$ is said to be a \emph{chain} if $N(v_i) \subseteq N(v_j)$ whenever $i \leq j$, where $N(v) := \{u \mid E(u,v)\}$ denotes the neighbourhood of $v$.  A bipartite graph $(A\cup B, E)$ is called a \emph{$k$-chain graph} if each of the two parts $A$ and $B$ can be further partitioned into at most $k$ chains.  A \emph{bichain graph} is a 2-chain graph.

\begin{figure}[ht]
    \centering
    \includegraphics[scale=0.75]{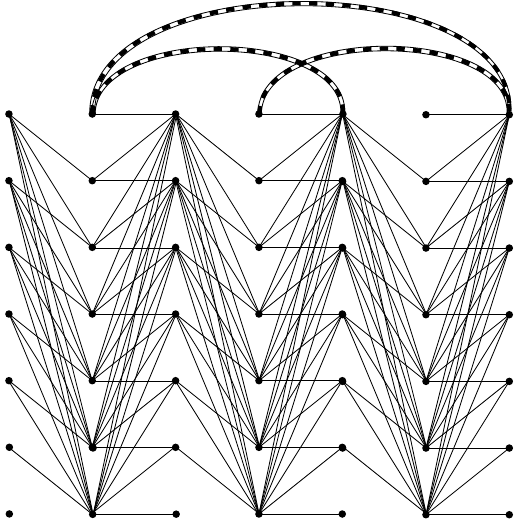}
    \caption{The bichain graph $Z_7$. A dashed line between any 2 columns denotes that the subgraph induced by the vertices of the two columns is a complete bipartite graph.}
    \label{fig:bichain}
\end{figure}

We now describe the bichain graph $Z_{n}$ as defined in~\cite{ABLS15}. The graph has vertex set $\{z_{i, j} \mid 1 \leq i \leq n, 1 \leq j \leq n\}$ (which can thus be seen as an $n \times n$ grid of points), and $\{z_{i, j}, z_{i', j'}\}$ is an edge if, and only if, one of the following holds: (i) 
 $j$ is odd, $j' = j + 1$ and $i < i'$; (ii) $j$ is even, $j' = j + 1$ and $i' \leq i$; or (iii)
 $j$ is even, $j'$ is odd and $j' \ge j + 3$.  The graph $Z_7$ is depicted in Figure~\ref{fig:bichain}.  The graph~$Z_n$ is \emph{$n$-universal} in that all bichain graphs on at most $n$ vertices are induced subgraphs of $Z_{n}$.  
It follows that the class {\bichain} of all bichain graphs is exactly the hereditary closure of the class $\{Z_{n} \mid n \ge 1\}$. 

Again, the grid structure is implicit in the graph $Z_n$.  What we show is that when $Z_n$ is expanded with unary relations for the bottom row  $\{z_{n,j} \mid 1\leq j \leq n\}$ and the last column $\{z_{i,n} \mid 1\leq i \leq n\}$, we can construct an FO interpretation to a class that contains the class $\cl{S}_\alpha$ for $\alpha = (23)^\omega$ in its hereditary closure.  We describe below the construction of this interpretation that we denote $\Psi$.

Let $H_n$ denote the expansion of $Z_n$ with unary predicates $\mathsf{bottom}$ and $\mathsf{last}$ that are respectively interpreted as the bottom row of $Z_n$ (namely the set $\{z_{n, j} \mid 1 \leq j \leq n\}$) and the last column of $Z_n$ (so the set $\{z_{i, n} \mid 1 \leq i \leq n\}$). Towards the construction of $\Psi$, we need the auxiliary predicates $\samecol{x, y}$ and $\adjcol{x, y}$. The first of these is true of $x, y$ in $H_n$ if, and only if, $x$ and~$y$ appear in the same column of $H_n$ and neither is the bottom element of that column. The second of these is true if $x$ and $y$ are in adjacent columns in $H_n$ and neither is the bottom element of its column.  We assume below that $n \ge 3$.

\[
\begin{array}{lll}
    \samecol{x, y} & := & \neg (\bottom{x} \vee \bottom{y}) ~\wedge \\
    & & \forall z (\bottom{z}\rightarrow (E(x, z) \leftrightarrow E(y, z)))\\[3pt]
    \adjcol{x, y} & := & \neg \big(\bottom{x} \vee \bottom{y}) ~\wedge \\ 
	& & ~~~\exists u \exists v (\samecol{u, x} \wedge \samecol{v, y} \wedge E(u, v)) ~\wedge\\
	& & ~~~\exists u \exists v (\samecol{u, x} \wedge \samecol{v, y} \wedge \neg E(u, v)\big)
\end{array}
\]
We briefly reason the correctness of the above predicate definitions. For \linebreak $\samecol{x, y}$, it is clear that this formula is true for any $x$ and $y$ in the same column of $H_n$ as long as they are not bottom elements.  To see that no other pair satisfies the formula, let $x = z_{i,j}$ and $y=z_{i',j'}$ with $j < j'$.  We argue by cases.   If $j'$ is odd, then $y$ is adjacent to the bottom element $u$ of column $j'+1$.  Moreover, since $j'+1$ is then even, $u$ is not adjacent to any $ z_{i,j}$ with $j < j'$.  On the other hand, if $j'$ is even, then we consider whether $j$ is odd or even.  If $j$ is odd, $x$ is adjacent to the bottom element of column $j+1$ and $y$ is not while if $j$ is odd, $x$ is adjacent to the bottom element of column $j'+1$ and~$y$ is not.

For $\adjcol{x, y}$, if $x$ and $y$ are in adjacent columns and not bottom elements of their respective columns, then let $x'$ and $x''$ be the elements in the column of $x$ in resp. the top row and the row just before the bottom in $H_n$, and let $y''$ be the  element corresponding to $x''$ in the column of $y$. Since $n \ge 3$, we have that $x', x'', y''$ are all distinct. We now see that if the column of $x$ is odd, then $E(x'', y'')$ is false while $E(x', y'')$ is true in $H_n$; otherwise, $E(x'', y'')$ is true while $E(x', y'')$ is false in $H_n$. Then $\adjcol{x, y}$ is true in $H_n$.  To see that no other pairs $(x, y)$ other than those just considered satisfy $\adjcol{x, y}$, let $x = z_{i,j}$ and $y=z_{i',j'}$ with $j < j'+1$.  We again argue by cases. If $j$ is odd, then since the only column $k \ge j$ for which some vertex of column $j$ is adjacent to some vertex of column $k$, is the column $k = j+1$, it follows that $\adjcol{x, y}$ is false. If $j$ is even, then if $j'$ is even, then no vertex of column $j$ is adjacent to any vertex of column $j'$, and if $j'$ is odd, then every vertex of column $j$ is adjacent to every vertex of column $j'$. In either case, $\adjcol{x, y}$ is false.

Consider now the interpretation $\Psi = (\Psi_V, \Psi_E)$ defined as:
\[
\begin{array}{lll}
    \Psi_V(x) & := & \neg \bottom{x} \wedge \neg \last{x} \\
    \Psi_E(x, y) & := & \adjcol{x, y} \wedge E(x, y) 
\end{array}
\]
It is easy to see that $\Psi(H_{2n+1})$ is the graph $Z_{2n}$ with the edges connecting non-adjacent columns removed; call this graph $Z_{2n}'$. Let $z_{i, j}$ be the vertex of~$Z_{2n}'$ in row $i$ and column $j$ (in  the natural grid in which the vertices of $Z_{2n}'$ are arranged). Consider the subgraph of $Z_{2n}'$ induced by the set $\mathsf{V}$ of vertices given by $\mathsf{V} = \{z_{i, j} \mid k+1 \leq i \leq k+n~\mbox{where}~\lfloor \frac{j}{2}\rfloor = k, 1 \leq j \leq n\}$. One verifies that this subgraph is indeed isomorphic to the graph $Y_n$ that is induced by the vertices in the first $n$ rows and first $n$ columns, in the graph $\mc{P}_\alpha$ where $\alpha = (23)^\omega$. 

\begin{proof}[Proof of Theorem~\ref{theorem:reductions}(\ref{bichain-graphs})]
Consider the interpretation $\Psi$ as described above (having parameters $\mathsf{bottom}$ and $\mathsf{last}$). The hereditary closure of $\Psi(\bichain)$ contains the hereditary closure of $\{Y_n \mid n \ge 1\}$. The latter class ($\{Y_n \mid n \ge 1\}\downarrow$) is nothing but $\cl{S}_\alpha$ for $\alpha = (23)^\omega$. We are now done by Theorem~\ref{theorem:word-defined-classes} and Proposition~\ref{prop:reduction-helper}.
\end{proof}

\noindent \textbf{Split permutation graphs.}
Recall that a \emph{split graph} is a graph $G$ whose vertex set can be partitioned into two sets $C$ and $I$ such that $C$ induces a clique in $G$ and $I$ is an independent set in $G$. A \emph{permutation graph} is a graph whose vertices represent the domain of a permutation, and each of whose edges determines an inversion in the permutation. Following~\cite{ABLS15}, we use the following characterization of split permutation graphs.

\begin{prop}[{\cite[Prop.~2.3]{ABLS15}}]\label{prop:split-perm-and-bichain-graphs}
Let $G$ be a split graph given together with a partition of its vertex set into a clique $C$ and an independent set $I$. Let $H$ be the bipartite graph obtained from $G$ by deleting the edges of $C$. Then $G$ is a split permutation graph if, and only if, $H$ is a bichain graph.
\end{prop}

\begin{figure}[ht]
    \centering
    \includegraphics[scale=0.75]{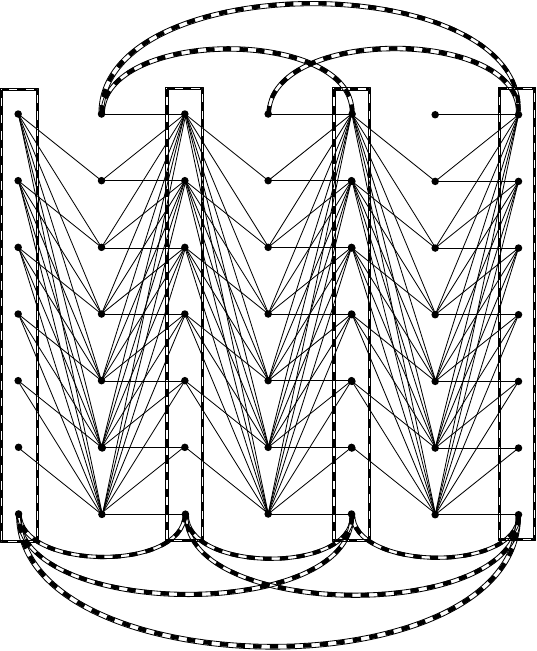}
    \caption{A split permutation graph $G$. The dashed line surrounding any odd column denotes that the vertices of the column form a clique, and a dashed line between two columns denotes that every vertex of one column is adjacent to every vertex of the other column. One sees that the vertices of the even columns form an independent set in $G$, while those of the odd columns form a clique. Deleting the edges in this clique gives us the bichain graph $Z_7$ depicted in Figure~\ref{fig:bichain}.}
    \label{fig:split-perm}
\end{figure}

Let $G$ be a split permutation graph with $(C, I)$ being a partition of its vertex set into a clique $C$ and an independent set $I$. Let $G^*$ be the expansion of $G$ with a unary predicate $P$ which is interpreted as the set $C$. Consider the $\fo$ interpretation $\Psi$ which removes from $G^*$ all edges inside $P$.
It is easy to see that $\Psi(G^*)$ is a bichain graph by Proposition~\ref{prop:split-perm-and-bichain-graphs}.

Let $\Psi$ be the FO interpretation as described above and $\textsf{SP}$ be the class of split permutation graphs. Then $\Psi(\textsf{SP})$, and hence its hereditary closure, contains the class $\bichain$. We are then done by Theorem~\ref{theorem:reductions}(\ref{bichain-graphs}) and Proposition~\ref{prop:reduction-helper}.

\vspace{3pt} \noindent \textbf{Bipartite permutation graphs.}
These graphs  are graphs that are bipartite as well as being permutation graphs. 
For our purposes, the following characterization is useful.
Consider the graph $P_{n}$ on vertex set $\{v_{i,j} \mid 1 \leq i,j \leq n \}$ where the only edges are between $v_{i,j}$ and $v_{i+1,j'}$ for $j' \leq j$\label{bip-perm}.  Then, the class of bipartite permutation graphs is exactly the  hereditary closure of the class $\{P_n \mid n \ge 1\}$~\cite{Lozin11}.  Now, it is easily seen that this class is exactly the class $\cl{S}_\alpha$ as described in Section~\ref{section:word-defined-classes}, for $\alpha = 2^\omega$, and this has been observed in~\cite{CollinsFKLZ18}.  Thus, Theorem~\ref{theorem:reductions}(\ref{BP-graphs}) follows from Theorem~\ref{theorem:word-defined-classes}.

\begin{figure}[ht]
    \centering
    \includegraphics[scale=0.7]{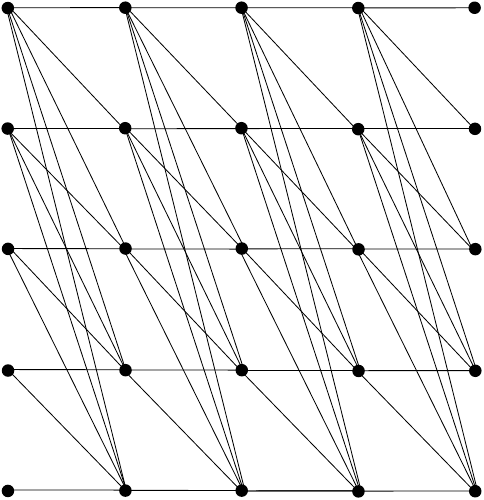}
    \caption{The bipartite permutation graph $P_5$.  In this illustration the vertex $v_{i,j}$ appears in row $j$ and column $i$.}
    \label{fig:bip_perm}
\end{figure}

\renewcommand{\odd}{\ensuremath{\mathsf{odd}}}
\newcommand{\xor}{\ensuremath{\oplus}}

\renewcommand{\odd}[1]{\ensuremath{\mathsf{odd}(#1)}}

\section{Power Graphs}\label{section:power-graphs}

\begin{figure}[ht]
    \centering
    \includegraphics[scale=0.63]{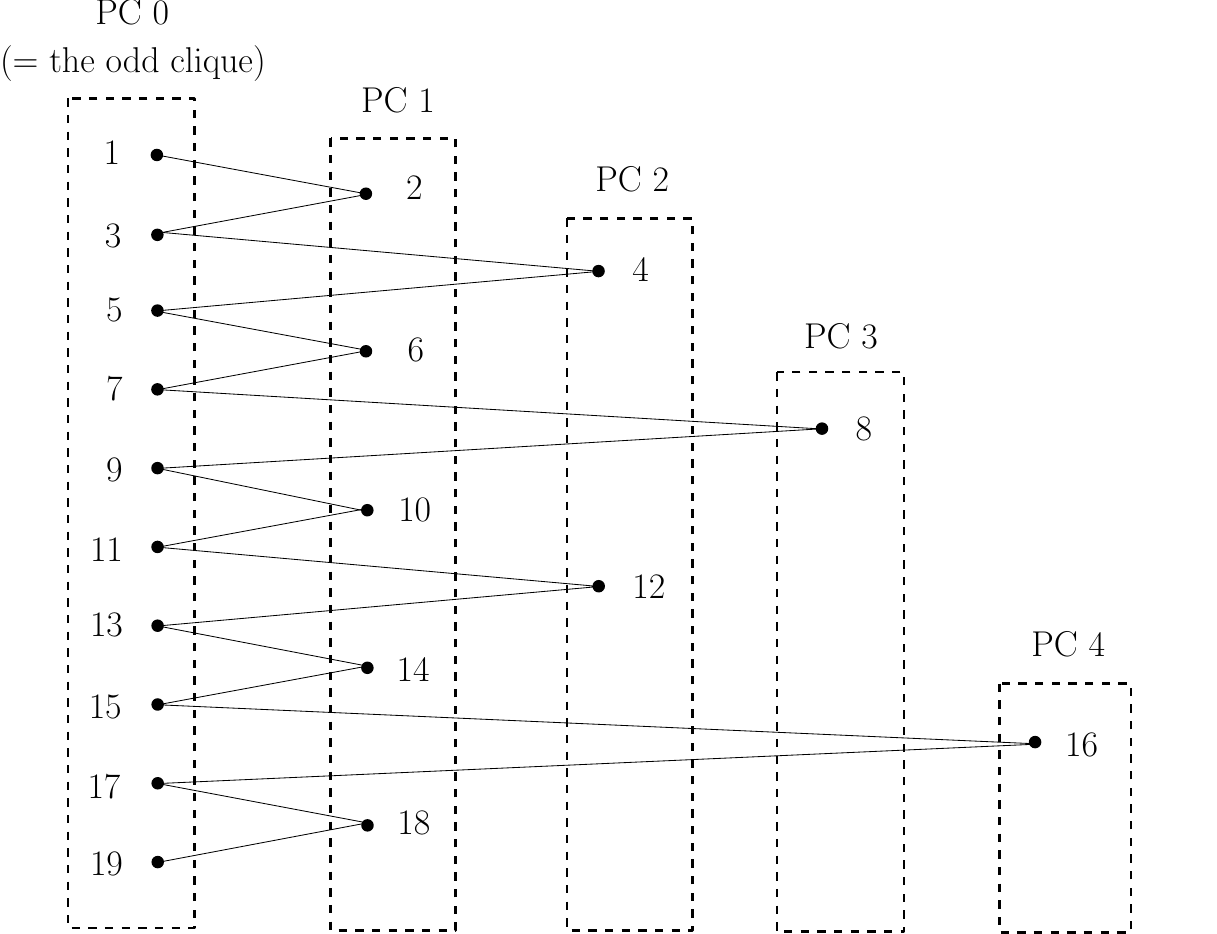}
    \caption{The power graph $D_{19}$ with ``PC $i$'' denoting the power clique corresponding to $i$.}
    \label{fig:power-graphs}
\end{figure}

In this section, we consider the class of \emph{power graphs} as  defined in~\cite{LRZ15} in the context of well-quasi-ordering  and clique-width.
Most of the classes that we have seen so far can be shown to not be well-quasi ordered under the induced subgraph relation.  In particular, all word-defined classes, unit interval graphs and bipartite permutation graphs can be seen to contain the antichain $\{ I_n \mid n \ge~1\}$ described after Proposition~\ref{prop:hereditary}.  We do not know whether  bichain graphs and split permutation graphs are well-quasi ordered, though it has been shown that their expansion with two labels is not a well-quasi ordered class~\cite{ABLS15}.  In contrast, power graphs constitute a class of graphs that is $\hucw$, that \emph{is} well-quasi ordered~\cite{LRZ15} and, as we show, is a minimal $\hucw$ class.  It was introduced precisely to demonstrate an  $\hucw$  class that is well-quasi ordered.  Minimality follows from arguments contained in~\cite{LRZ15}, but was not observed there.  We now define the class of power graphs. We  show that they are minimal and then in the remainder of the section show that they admit interpretability of grids.

For $n \ge 1$, we define the graph $D_n$ as follows.
The vertex set of $D_n$ is $\left[n\right] = \{1, \ldots, n\}$.  For each $i < n$, there is an edge between $i$ and $i+1$---we call these  \emph{path edges}.  Furthermore, there is an edge between $i$ and $j$ if the largest power of 2 that divides $i$ is the same as the largest power of 2 that divides~$j$---we call these \emph{clique edges}.  To understand this terminology, note that we can see~$D_n$ as consisting of a simple path of length $n$, along with, for each $k$ such that $2^k \leq n$, a clique on all vertices $j = 2^k \cdot (2r + 1)$ for some $r \ge 0$---we call this the \emph{power clique} corresponding to $k$.  In particular, taking $k=0$, there is a clique formed by all the odd elements, which we call the \emph{odd clique}.   An example graph is illustrated in Figure~\ref{fig:power-graphs}.  Observe that the path edges, which are the only edges with endpoints in different power cliques always have one end point in the odd clique, and one outside it.
The class of power graphs, denoted {\powergraphs}, is now defined as the hereditary closure of the class $\{ D_n \mid n \ge 1 \}$.

\subsection{Minimality of Power Graphs} \label{subsection:power-graph-minimality}
\begin{prop}\label{prop:minimality-of-power-graphs}
The class {\powergraphs} is a minimal hereditary class of unbounded clique-width.
\end{prop}
That $\powergraphs$ is a hereditary class of unbounded clique-width has already been shown in~\cite{LRZ15}.   Thus, we only need to show that no proper subclass has this property.

Given a graph $G \in \powergraphs$ which is a subgraph of $D_n$, define an \emph{interval} in $G$ to be a set $S \subseteq [n]$ of vertices of $G$ such that if $i,j \in S$ with $i < j$ and $k$ is a vertex of $G$ with $i< k < j$ then $k \in S$.  We call a subgraph of $G$ induced by an interval a \emph{factor} of $G$.  We now recall the following two results proved in~\cite{LRZ15}.

\begin{lemma}[Lemma 11,~\cite{LRZ15}]\label{lemma:induced-subgraph-of-factor}
Let $G$ be a graph in $\powergraphs$. Then there exists an integer $t = t(G)$ such that for any $n \ge t$, every factor of $D_n$ of length at least $t$ contains $G$ as an induced subgraph.
\end{lemma}

\begin{theorem}[Theorem 2,~\cite{LRZ15}]\label{theorem:factor-length-and-cwd}
Let $G$ be a graph in $\powergraphs$ such that the length of the longest factor in $G$ is $t$.  Then the clique-width of $G$ is at most $2(\log t + 4)$.
\end{theorem}

\begin{proof}[Proof of Proposition~\ref{prop:minimality-of-power-graphs}] Consider a proper hereditary subclass $\cl{S}$ of \linebreak $\powergraphs$; then $\cl{S}$ excludes a graph $G \in \powergraphs$. Let $t = t(G)$ be as given by Lemma~\ref{lemma:induced-subgraph-of-factor}.
Let $\cl{S} = \cl{S}_1 \cup \cl{S}_2$ where $\cl{S}_1 = \cl{S} \cap \{D_n \mid n < t\}\!\!\downarrow$ and $\cl{S}_2 = \cl{S} \cap \{D_n \mid n \ge t\}\!\!\downarrow$. Observe that $\cl{S}_1$ has finitely many graphs up to isomorphism.

We show that for each  $X \in \cl{S}_2$, every factor of $X$ has length $< t$. For otherwise~$X$ has a factor $Y$ of length $\ge t$ and there is $p \ge 1$ such that $X \subseteq~D_p$ and so $Y$ is also a factor of $D_p$.  Hence by Lemma~\ref{lemma:induced-subgraph-of-factor}, we have $G$ is an induced subgraph of $Y$, whereby it is also an induced subgraph of $X$.  Since $\cl{S}$ is hereditary, $G \in \cl{S}$ which is a contradiction.

By Theorem~\ref{theorem:factor-length-and-cwd}, every $X \in \cl{S}_2$ has clique-width $\leq k = 2(\log t + 4)$. Then $\cl{S}_2$ has bounded clique-width, and hence so does $\cl{S}$ since $\cl{S}_1$ is finite.
\end{proof}

\subsection{Interpreting grids in Power Graphs}\label{sec:power-grids}
We now establish the main result of this section, showing that power graphs do not provide a counter-example to Seese's conjecture.

\begin{theorem}\label{theorem:power-graphs}
There exists an $\mso$ interpretation $\Theta$ such that  $\Theta(\powergraphs)$ contains all square grids. 
\end{theorem}

We show Theorem~\ref{theorem:power-graphs} by showing that there exists an $\mso$ interpretation~$\Phi$ such that the hereditary closure of $\Phi(\powergraphs)$ contains all bipartite permutation graphs. We are then done by Theorem~\ref{theorem:reductions} and Proposition~\ref{prop:reduction-helper}.  Indeed, it suffices to show that we can interpret grids in a subset of $\powergraphs$, and we do this for the set $\{D_n \mid n \in \mathbb{N} \text{ even and } n > 9 \}$. We first describe the overall ideas involved in the construction of $\Phi$, and provide the details subsequently.
 

We first show that there exists an FO formula
$\odd{x}$ such that if $x$ is a number in $D_n$ with $n \ge 9$, then
$\odd{x}$ is true if, and only if, $x$ is an odd number. With this formula at hand, we can distinguish path edges from clique edges.  Indeed, an edge is a path edge if, and only if, it has exactly one end point that is odd.  In $D_n$, the path edges form a  simple path of length $n-1$ and, if $n$ is even, then only one of the two end points satisfies $\odd{x}$.  This allows us to give this simple path an orientation: for each path edge $(x,x+1)$ we can identify the direction $x \rightarrow x+1$.  The transitive closure of this relation (which is definable in $\mso$), gives us a definition of the natural linear order on $D_n$.

Once we have defined a linear order $\leq$ on $D_n$, this induces a linear order on the power cliques: namely, a clique $C$ is below $C'$ if the $\leq$-minimal element of $C$ is less than the $\leq$-minimal element of $C'$.  Indeed, we can also define a successor relation on cliques from this.  From these, we define a relation that relates a pair $x$ and $y$ precisely if $y$ occurs after $x$ in the linear order $\leq$ and occurs in the power clique that is successor to the power clique containing $x$.  It is easy to see that the graph induced by this relation contains arbitrarily large bipartite permutation graphs $P_k$ as defined on page~\pageref{bip-perm}.  
  
We now give the details of the construction described above. In addition to $\odd{x}$, we need a number of auxiliary predicates along the way. 

\newcommand{\clique}{\textsf{clique}}
\newcommand{\pathedge}{\textsf{pathedge}}
\renewcommand{\xor}{\oplus}
\newcommand{\mybetween}{\textsf{between}}

\begin{enumerate}
\item 
We first define the FO formula $\odd{x}$.
\[
\begin{array}{lll}
\odd{x} & := & \exists y \exists z \exists w \big( ``x, y, z, w
~\text{form a 4-clique except for the}~ z-w~\text{edge}"\big)\\
\end{array}
\]

It is easy to see that for $n \ge 9$, all odd numbers in $D_n$ satisfy $\odd{x}$.  If $x$ is odd with $x < n-3$, this is witnessed by $y = x + 2$, $w = x + 4$ and $z = x + 1$, otherwise by $y = x - 2$, $w = x - 4$ and $z = x - 1$.

To show that the even numbers of $D_n$ do not satisfy $\odd{x}$, first observe that in any power clique other than the odd clique, since the numbers in the clique are of the form $2^k \cdot (2r + 1)$ for fixed $k$, the difference between any two numbers in the clique is at least $2^{k+1}$, which is at least $4$ since $k \ge 1$.  Suppose now that $x$ is an even number in $D_n$ and $x,y,z$ form a $3$-clique.  We argue that any $w$ that is adjacent to both $x$ and $y$ must also be adjacent to $z$ showing that $\odd{x}$ is not satisfied.  Consider the two cases:

\begin{itemize}
\item The edge between $x$ and $y$ is a clique edge.  Then $|x - y| \ge 4$.  If $z$ is in a different power clique, then $|x - z| = 1$ and $|z - y| = 1$, whereby $|x - y| \leq 2$ -- a contradiction.  Thus $z$ is in the same power clique as~$x$ and $y$.  By the same argument, $w$ is the same power clique as $x$ and~$y$, so there is a clique edge $z-w$.

\item The edge between $x$ and $y$ is a path edge and so $|x - y| = 1$. Then the edges from $z$ to $x$ and $y$ cannot both be path edges, as you cannot have a triangle of such edges.  So, one of them is a clique edge. If $z$ is in the same power clique as $x$, then $|x - z| \ge 4$ and $|y - z| = 1$, which is a contradiction, so $z$ must be in the same power clique as $y$.  By the same argument, $w$ is in the same clique as $y$, so there is a clique edge $z-w$.
\end{itemize}

\begin{remark}
The formula $\odd{x}$ is central to our construction below and we  assume henceforth that $n \ge 9$.
\end{remark}

\item 

$\clique(x, y)$ and $\pathedge(x, y)$: The formula
$\clique(x, y)$ is true of the pair $(x, y)$ in $D_n$ if, and only if, $x$ and $y$
are in the same power clique.  The formula $\pathedge(x, y)$ is true
if, and only if, $|x - y| = 1$. 
\[
\begin{array}{lll}
\pathedge(x, y) & := & E(x, y) \wedge (\odd{x} \xor \odd{y})\\
\clique(x, y) & := & E(x, y) \wedge \neg \pathedge(x,y)
\end{array}
\]

\item $\path{P, x, y}$: This predicate is true of all
triples $(P, x, y)$ for an $\mso$ variable $P$ and $x, y \in D_n$ if $P$ is the (unique) path between $x$ and $y$, whose edges are all path edges. Below $\exists! w$ denotes ``there is a unique $w$ such that...''.
\[
\begin{array}{lll}
  \path{P, x, y} & := & (P(x) \wedge P(y) \bigwedge \\
  & & \exists! w (P(w) \wedge \pathedge(x, w)) \bigwedge\\
  & & \exists! w (P(w) \wedge \pathedge(y, w)) \bigwedge\\
  & & \forall w \big((P(w) \wedge w \neq x \wedge w \neq y) \rightarrow \\
  & & ~~~~~~\exists u \exists v (P(u) \wedge P(v) \wedge \pathedge(u, w) ~\wedge\\
  & & ~~~~~~~~~~~~~~\pathedge(v, w) \wedge u \neq v)\big)\\
\end{array}
\]

\item $\mybetween(x, y, z)$: This predicate is true of all
triples $(x, y, z)$ in $D_n$ such that $y$ appears somewhere along the (unique) path between $x$ and $z$ ($y$ could be one of $x$ or $z$).
\[
\begin{array}{lll}
  \mybetween(x, y, z) & := & \exists P (\path{P, x, z} \wedge P(y))\\
\end{array}
\]
\end{enumerate}

We now make a few observations about $D_n$.  Note that since the path edges are definable, and they form a simple path from $1$ to $n$, the only possible automorphisms are the trivial one and the map that reverses the order, in particular mapping $n$ to $1$.  Moreover, since the odd numbers are definable, for the order reversing map to be an automorphism, $n$ must be odd.  We can say more: a more careful analysis shows that the order reversal preserves all power cliques if, and only if, $n=2^k-1$ for some $k$.  However, for our purposes it suffices to note that whenever $n$ is even $D_n$ has no non-trivial automorphisms.  The predicates we define next are for \emph{even $n > 9$}.

\renewcommand{\one}{\textsf{one}}
\newcommand{\mysucc}{\mathsf{succ}}
\newcommand{\cliquemin}{\textsf{cliquemin}}
\newcommand{\cliqueminlinord}{\textsf{cliqueord}}
\newcommand{\cliqueminsucc}{\textsf{cliquemin-succ}}

\begin{enumerate}
\setcounter{enumi}{4}
    
\item $\one(x)$: This predicate is satisfied by $x$ in $D_n$
if, and only if, $x = 1$.  It defines the unique (when $n$ is even) odd element that has only one path edge incident on it.
\[
\begin{array}{lll}
  \one(x) & := & \odd{x} \wedge \neg \exists z_1 \exists z_2 (\pathedge(x, z_1) \wedge \pathedge(x, z_2) \wedge z_1 \neq z_2) 
\end{array}
\]
This now allows us to orient the path edges to obtain the natural successor relation on $D_n$.
\item $\mysucc(x, y)$: This predicate is satisfied by $x, y$ in $D_n$ if, and only if, $y = x+1$.
\[
\begin{array}{lll}
  \mysucc(x,y) & := & \pathedge(x,y) \wedge  \exists z ( \one(z) \wedge \mybetween(z,x,y))
\end{array}
\]
As usual, we can then define in $\mso$ a formula $\linord{x,y}$ which defines the reflexive and transitive closure of $\mysucc$.

\item $\cliquemin(x)$: This predicate is true of $x$ in $D_n$
if, and only if, $x$ is the minimum element of its power clique (i.e. $x = 2^k$
for some $k \ge 0$).
\[
\begin{array}{lll}
  \cliquemin(x) & := &  \forall y  (\clique(x, y) \rightarrow \linord{x,y}). \\
\end{array}
\]

The linear order defined by $\mathsf{linord}$ then allows us to linearly order the power cliques.

\item $\cliqueminlinord(x, y)$: This predicate is true of the
pair $(x, y)$ in $D_n$ if, and only if, it is the case that the minimum element in the power clique of $x$ is less than the minimum element in the power clique of $y$.
\[
\begin{array}{lll@{}l}
  \cliqueminlinord(x, y) & := & 
  \exists z_1 z_2 & (\cliquemin(z_1) \wedge \cliquemin(z_2) \wedge \clique(x,z_1) \\ 
  & & &  \wedge ~\clique(y,z_2) \wedge \linord{z_1,z_2})
\end{array}
\]

This ordering of the power cliques and the fact that $\mathsf{linord}$ linearly orders each clique gives us sufficient structure to define arbitrarily large grids.  To see this concretely, consider the following relation.

\item $\cliqueminsucc(x, y)$: This predicate is true if $x$ is in the power clique corresponding to $k$ and $y$ in the power clique corresponding to $k+1$ for some $k$.  
\[
\begin{array}{lll}
  \cliqueminsucc(x, y) & := & \neg \clique (x,y) \wedge \cliqueminlinord(x, y) \wedge \\
  & & \forall z (\cliqueminlinord(x, z) \rightarrow \\
  & & ~~~~~~(\clique(z,y) \vee \cliqueminlinord(y, z)))\\
\end{array}
\]

\item Consider now the relation $\mathsf{forward}(x,y)$ defined by
\[
\begin{array}{lll}
  \mathsf{forward}(x,y) & := & \cliqueminsucc(x,y) \wedge \linord{x,y}.
\end{array}
\]
This relates an element $x$ in the power clique corresponding to $k$ to all elements of the power clique corresponding to $k+1$ that are greater than~$x$.
\end{enumerate}

Then, the interpretation $\Phi = (\Phi_V(x), \Phi_E(x, y))$  given by 
\[
\begin{array}{lll}
    \Phi_V(x) & := & \true  \\
    \Phi_E(x, y) & := & \mathsf{forward}(x, y) \vee  \mathsf{forward}(y,x)   \\
\end{array} 
\]
maps $D_n$ to a graph whose edge relation is the symmetric closure of $\mathsf{forward}$.  We claim that the graph $\Phi(D_n)$ contains a large bipartite permutation graph as an induced subgraph.  To see this, choose the largest value $k$ such that the power clique corresponding to $k$ contains at least $k$ elements in $D_n$ (in other words $2^k(2k-1) \leq n$).  Consider the subgraph of $D_n$ induced by the set of vertices $\{v_{i,j} \mid 0\leq i,j \leq k-1\}$ where $v_{i,j} = j\cdot 2^{k+1} + 2^{i+1}$.  Each $v_{i,j}$ is then in the power clique corresponding to $i+1$ and it is easily checked that there is an edge between $v_{i,j}$ and $v_{i',j'}$ in $\Phi(D_n)$ precisely when $i' = i+1$ and $j \leq j'$.

\begin{proof}[Proof of Theorem~\ref{theorem:power-graphs}]
  As established above, the graph $\Phi(D_n)$ contains an induced subgraph isomorphic to the bipartite permutation graph $P_k$ as long as $2^k(2k-1) \leq n$, and $n$ is even and at least $9$.
  Then the hereditary closure of $\Phi(\powergraphs)$ contains all bipartite permutation graphs, whereby, by Theorem~\ref{theorem:reductions} and Proposition~\ref{prop:reduction-helper}, we are done.
\end{proof}

\section{Conclusion}\label{section:conclusion}
The study of monadic second-order logic on graphs has attracted great attention in recent years.  An important aspect of work on this logic is to identify classes of graphs on which $\mso$ is well behaved.  Seese's conjecture is an important focus of this classification effort.  In its stronger form it offers a dichotomy: any class of graphs is either interpretable in trees and therefore has bounded clique-width and is well-behaved \emph{or} it interprets arbitrarily large grids and its $\mso$ theory is then undecidable.  

We show that Seese's conjecture could be established by considering two kinds of graph classes: minimal hereditary classes of unbounded clique-width and antichains of unbounded clique-width.  Showing that all such classes interpret unbounded grids would suffice.  While we do not have a complete taxonomy of such classes, we investigated all the ones known and showed that none of them provides a counter-example to Seese's conjecture.

We know of only two explicit constructions of antichains of unbounded clique-width: the one presented in this paper and the one due to Korpelainen~\cite{Korpelainen16}.  Both are explicitly based on grids and easily admit an interpretation of arbitrarily large grids.  On the other hand, there is a richer landscape of known minimal HUCW classes and we explore this systematically.

One could weaken the strong conjecture by requiring only that the classes of unbounded clique-width admit \emph{$\mso$ transductions} of grids, rather than interpretations (see~\cite{courcelle-engelfriet} for a discussion of transductions).  This would still suffice to establish Seese's conjecture.  In all the cases we consider, however, we establish the stronger form, i.e.\ an interpretation of grids.

It is also worth pointing out that for many of the classes we consider, the original proofs that they have unbounded clique-width require sophisticated bespoke arguments.  The interpretation of grids in the classes also provides a uniform method of proving that they have unbounded clique-width.

As a final remark, it is worth noting that there are standard graph operations which allow us to construct new minimal $\hucw$ graph classes from the ones we have.  For example, taking the graph complement of all graphs in a class $\cl{C}$ yields a class that is also minimal $\hucw$ if $\cl{C}$ is.  Since this operation is itself an $\mso$ interpretation, the results about interpreting arbitrarily large grids apply to the resulting classes as well.

\section*{Acknowledgements}
Research supported by the Leverhulme Trust through a Research Project Grant on ``Logical
Fractals"

\bibliography{refs}

\end{document}